\documentclass[review]{elsarticle}
\usepackage{lineno,hyperref}
\usepackage{soul, framed}
\journal{Naval Research Logistics(NRL)}

\biboptions{authoryear}

\usepackage{amssymb}
\usepackage{amsmath}
\usepackage{mathrsfs}
\usepackage{graphicx}
\usepackage{booktabs}
\usepackage{setspace}
\usepackage{url}
\usepackage{makecell}
\usepackage{amsthm}
\usepackage{amstext}
\usepackage{multirow}
\usepackage{float}
\usepackage{setspace}
\usepackage{color}
\usepackage{soul}
\usepackage[linesnumbered,ruled]{algorithm2e}
\usepackage{changepage}

\newtheorem{theorem}{Theorem}
\newtheorem{lemma}{Lemma}

\newproof{pot}{Proof of Theorem \ref{thm2}}

\begin{document}
\begin{frontmatter}
	\title{Serial-batch scheduling to minimise the total weighted late work}

	\author[a]{Yao-Wen Sang}
	\author[a]{Naiming Xie}
	\author[a]{Jian Chen\corref{cor}}
	\ead{jchen@nuaa.edu.cn}
	\cortext[cor]{Corresponding author}
	\author[b]{Malgorzata Sterna}
	\author[b]{Jacek Blazewicz}

	\address[a]{College of Economics and Management, Nanjing University of Aeronautics and Astronautics, Nanjing, Jiangsu 210016, PR China}
	\address[b]{Institute of Computing Science, Poznan University of Technology, Piotrowo 2, Poznan 60–965, Poland}
	\begin{spacing}{1.31}
		\begin{abstract}
			We study the problem of scheduling jobs on a serial-batch machine with the aim of minimising the total weighted late work. In a serial-batch setting, jobs within a batch are processed sequentially, and none are removed from the machine until the last job in the batch completes its processing. The processing time of a batch is the sum of the processing times of the jobs within it, and the completion time for each job in the batch is equal to the makespan of the jobs in the batch. When a new batch begins, a constant setup time is required for the machine. We show that minimising the total weighted late work in this environment is $NP$-hard even if all jobs have a common due date and unit weight.
			For the general problem, we present a pseudo-polynomial time dynamic programming algorithm.
			Additionally, we explore two special cases, i.e., one with a common due date and another with an agreeable condition among due dates, processing times and weights. For both special cases, we develop specialised pseudo-polynomial time dynamic programming algorithms. The proposed approaches are equipped with specialised acceleration techniques to enhance their computational performance. The extended experiments demonstrate that the dynamic programming algorithms outperform Gurobi in time efficiency.\\
		\end{abstract}

	\end{spacing}
\end{frontmatter}

\section{Introduction}

In recent years, serial-batch scheduling (cf. \citealp{gahm2022scheduling, uzunoglu2023machine, wu2023exact}) and late work-related scheduling (cf. \citealp{sterna2021late, chen2023revisit, justkowiak2023single, sang2023single, shabtay2023new}) have attracted significant attention, each independently.
In serial-batch setting, jobs within a batch are processed sequentially, and none are removed from a machine until the last job in the batch completes its processing. The processing time of a batch is determined by the sum of the processing times of the jobs within it, and the completion time for each job in the batch is equal to the completion time of the whole batch. Moreover, when a new batch begins, a constant setup time is required for the machine. Serial-batch scheduling finds widespread applications across various industries, including aluminium-making process \citep{pei2019single}, production management \citep{karakutuk2024two}, semiconductor manufacturing \citep{cerekci2010dynamic}, and drilling process of fuel injector production \citep{ozturk2022serial}.
In late work-related scheduling, the late work of a job is defined as the job processing time executed after its due date. Late work-related scheduling has been thoroughly explored in domains such as control systems \citep{blazewicz1984scheduling}, agriculture \citep{blazewicz2004open}, logistics industry \citep{chen2022pareto}, software development \citep{sterna2011survey}, and flexible manufacturing \citep{sterna2007late}.

This paper presents the first attempt to integrate serial-batch processing with late work minimisation. For this reason, we have focused on the theoretical studies of the newly proposed model, i.e., on establishing its computational complexity. Determining the complexity of the basic model, i.e., proving its polynomial time solvability or strong/week $NP$-hardness, is crucial for directing further research, particularly on more complex models reflecting various real world scenarios. The problem formulation builds on established concepts in scheduling theory and their motivations rooted in real-world applications, where sequential batching, setup operations, and delivery-critical performance are all integral. Such characteristics are commonly observed in industrial domains.

\begin{spacing}{1.4}
	Problems involving serial-batch processing and late work minimisation can be found in cold-chain logistics for example, in the batch picking and packing of chilled products for e-commerce grocery fulfilment or fresh food distribution. In these operations, customer orders are often grouped into batches, corresponding to a vehicle or a courier to which these orders are assigned. Batches are processed sequentially at temperature-controlled workstations. The entire batch must remain on the station until all items are picked and packed, and a fixed setup time is typically required between batches to clean equipment, reconfigure packaging materials for different product types, and preparing a courier for departure to customers. Each order is associated with a due date reflecting its latest acceptable departure time for delivery. If processing continues beyond this time, the portion of time consumed after the due date constitutes late work directly impacting product freshness, increasing spoilage risk, and causing missed delivery slots resulting in customer dissatisfaction. Since orders may differ in the value or priority, minimising the total weighted late work has a beneficial effect on ensuring cold-chain reliability, reducing operational losses, and maintaining customer satisfaction.

	Another motivating scenario is observed in aerospace composite manufacturing, where production efficiency and schedule stability are paramount. In the prepreg cutting stage, composite materials are first fixed onto a cutting table using laser alignment tools, followed by sequential contour cutting of each part. This process involves a setup time to position and secure the material before cutting begins, and the cutting table remains occupied until the full batch of parts is completed, which reflects a serial-batch processing structure. Each job's processing time depends on the complexity of its geometric outline, while due dates are typically constrained by the timing of downstream layup and curing operations. When cutting extends beyond a job's due date, the late portion of processing represents a direct delay to subsequent steps on tightly synchronized composite production lines. Moreover, prolonged exposure of uncured prepreg to ambient air increases the risk of moisture absorption, which may degrade the material's mechanical properties. Minimising the total weighted late work thus improves production throughput and helps prevent material deterioration in aerospace composite manufacturing.
\end{spacing}

As we have mentioned, serial-batch processing has not been studied in the context of late work minimisation so far. For this reason, we focused on the complexity analysis of the basic model. By using the three-field notation, the general problem considered in this paper is denoted as $1|s\text{-}batch|Y_w$, where $s\text{-}batch$ represents the serial-batch machine environment, and $Y_w$ denotes the total weighted late work. Additionally, we analyse two special cases of this problem with a common due date, $1|s\text{-}batch,d_j = d|Y_w$, as well as with agreeable conditions, $1|s\text{-}batch,d_j\uparrow p_j \uparrow w_j\downarrow|Y_w$, where $d_j \uparrow p_j\uparrow w_j\downarrow$ indicates the agreeable condition among processing times, due dates and weights, i.e., for any two jobs $J_i$ and $J_j$, if $d_i\geq d_j$, then $p_i\geq p_j$ and $w_i\leq w_j$.

The results available in the literature show that the complexity of serial-batching problems depends on the objective function used for evaluating the quality of schedules, and scheduling problems with late work minimisation are intractable even for a single machine. The most important theoretical results related to our research are collected in Table \ref{Tab: Overview of the complexity results}. \cite{coffman1990batch} considered the problem of minimising the total completion time, i.e., $1|s\text{-}batch|\sum C_j$. They showed that among all schedules, the one leading to the minimum total completion time has the jobs sequenced in non-decreasing order of processing times, and they constructed it in $O(n\log n)$ time. With regard to the objective of minimising the total weighted completion time, \cite{albers1993complexity} showed that problem $1|s\text{-}batch|\sum w_jC_j$ is strongly $NP$-hard. When due dates are involved, \cite{webster1995scheduling} solved the problem of minimising the maximum lateness, $1|s\text{-}batch|L_{\max}$, by proposing a dynamic programming (DP) algorithm with a time complexity of $O(n^2)$.
For the total tardiness minimisation problem $1|s\text{-}batch|\sum T_j$, \cite{du1990minimizing} proved that it is weakly $NP$-hard and
\cite{baptiste2001minimizing} provided a DP algorithm with running time $O(n^{11}\max(p_{\max}, s)^7)$, where $s$ denotes setup time and $p_{\max}$ denotes the maximum job processing time.
\cite{brucker1996single} solved the problem with the minimum number of late jobs criterion, $1|s\text{-}batch|\sum U_j$, in $O(n^3)$ time, and proved that the weighted problem $1|s\text{-}batch|\sum w_jU_j$ is weakly $NP$-hard and $O(n^2\sum w_j)$-time solvable.
For scheduling problems with late work criteria, \cite{potts1992single} proved that the problem of minimising the total late work, $1||Y$, is weakly $NP$-hard and $O(n\sum p_j)$-time solvable. \cite{hariri1995single} studied the weighted problem $1||Y_w$ and constructed a DP with $O(n^2 \sum p_j)$ time complexity. The special cases of both problems with a common due date, $1|d_j=d|Y$ and $1|d_j=d|Y_w$, are polynomially solvable in $O(n)$ and $O(n\log n)$ time by the algorithms proposed by \cite{potts1992single} and \cite{hariri1995single}, respectively.
As we have mentioned late work criteria have not been studied for serial-batching environment. They have been investigated for parallel-batching setting only ($p\text{-}batch$), where the batch completion time is determined by the maximum processing time among jobs assigned to this batch. For this model \cite{zhang2005np} showed that problem $1|p\text{-}batch|Y_w$ is weakly $NP$-hard, while \cite{ren2009np} proved that the problem without weights, i.e., $1|p\text{-}batch|Y$ is also weakly $NP$-hard. Both problems $1|p\text{-}batch|Y_w$ and $1|p\text{-}batch|Y$ can be solved by using the DP algorithm proposed by \cite{brucker1996single}.

The remainder of the paper is organized as follows. Section \ref{Sec: s-batch} provides the formal definition of the general problem, $1|s\text{-}batch|Y_w$, and shows that it is $NP$-hard. Section \ref{Sec: DP} presents a pseudo-polynomial time algorithm to solve $1|s\text{-}batch|Y_w$, which proves its weak $NP$-hardness. The following sections focus on special cases of this model where the computational burden of the dynamic programming can be reduced. Specifically, Section \ref{Sec: Common Due Date} provides a specialised algorithm for problem $1|s\text{-}batch, d_j = d|Y_w$, while Section \ref{Sec: Agreeable} focuses on problem $1|s\text{-}batch,d_j \uparrow p_j\uparrow w_j\downarrow|Y_w$.
For all the proposed algorithms, specialised acceleration techniques are introduced in their respective sections.  Section \ref{Sec: Experiments} presents the results of extensive computational experiments, where the efficiency of the proposed methods is evaluated against the Gurobi optimizer.
Conclusions and future research directions are given in Section \ref{Sec: Conclusion and Future Research}.

\vspace{-1em}
\begin{table}[H]
	\centering
	\caption{The complexity results of related problems}\label{Tab: Overview of the complexity results}
	\scalebox{0.6}{
		\begin{tabular}{cccc}
			\toprule
			\textbf{Classification} & \textbf{Problem}               & \textbf{Complexity}                              & \textbf{Reference}                   \\ \midrule
			                        & $1|s\text{-}batch|\sum C_j$    & $O(n\log n)$                                     & \cite{coffman1990batch}              \\
			                        & $1|s\text{-}batch|\sum w_jC_j$ & strongly $NP$-hard                               & \cite{albers1993complexity}          \\
			{serial-batch}          & $1|s\text{-}batch|\sum U_j$    & $O(n^3)$                                         & \cite{brucker1996single}             \\
			{scheduling}            & $1|s\text{-}batch|\sum w_jU_j$ & weakly $NP$-hard, $O(n^2\sum w_j)$               & \cite{brucker1996single}             \\
			                        & $1|s\text{-}batch|L_{\max}$    & $O(n^2)$                                         & \cite{webster1995scheduling}         \\
			                        & $1|s\text{-}batch|\sum T_j$    & weakly $NP$-hard, $O(n^{11}\max(p_{\max}, s)^7)$ & \cite{du1990minimizing}              \\
			\cline{1-4}
			                        & $1||Y$                         & weakly $NP$-hard, $O(n\sum p_j)$                 & \cite{potts1992single}               \\
			                        & $1|d_j=d|Y$                    & $O(n)$                                           & \cite{potts1992single}               \\
			{late work-related}     & $1||Y_w$                       & weakly $NP$-hard, $O(n^2 \sum p_j)$              & \cite{hariri1995single}              \\
			{scheduling}            & $1|d_j=d|Y_w$                  & $O(n\log n)$                                     & \cite{hariri1995single}              \\
			                        & $1|p\text{-}batch|Y_w$         & weakly $NP$-hard, $O(n^2 \sum p_j)$              & \cite{zhang2005np,brucker1996single} \\
			                        & $1|p\text{-}batch|Y$           & weakly $NP$-hard, $O(n^2 \sum p_j)$              & \cite{ren2009np,brucker1996single}   \\
			\bottomrule
		\end{tabular}}
\end{table}

\section{Problem Description}\label{Sec: s-batch}

Formally, the serial-batching problem with the total weighted late work, $1|s\text{-}batch|Y_w$, is described as follows. We are given a set of jobs $J=\{J_1,J_2,\ldots,J_n\}$. Each job $J_j$ has a processing time $p_j$, a due date $d_j$, and a weight $w_j$. All jobs have to be processed on a serial-batch machine. The jobs in the same batch start and complete simultaneously. The processing time of a batch equals to the sum of the processing times of its jobs. The maximum number of jobs in a batch is unlimited. Before starting a batch, a constant setup time $s$ is needed. Let $C_j(\sigma)$ be the completion time of job $J_j$ under a given schedule $\sigma$. The late work of job $J_j$ is equal to the part of the processing time scheduled after its due date, denoted as $Y_j(\sigma)=\min\{\max\{C_j(\sigma)-d_j,0\},p_j\}$. Let $Y(\sigma)$ and $Y_w(\sigma)$ be the total late work and total weighted late work under schedule $\sigma$, respectively, i.e., $Y(\sigma) = \sum_{j=1}^n Y_j(\sigma)$ and $Y_w(\sigma) = \sum_{j=1}^n w_j Y_j(\sigma)$. When there is no confusion, we omit $\sigma$ in $C_j(\sigma)$, $Y_j(\sigma)$, $Y(\sigma)$ and $Y_w(\sigma)$ using abbreviated forms $C_j$, $Y_j$, $Y$ and $Y_w$, respectively. To improve readability, Table \ref{Tab: notation} presents a summary of basic notations used in the paper.

\begin{table}[htbp]

	\renewcommand{\arraystretch}{1.3}
	\centering
	\caption{Basic notations}
	\label{Tab: notation}
	\resizebox{\textwidth}{!}{%
		\begin{tabular}{p{3.5cm} p{9.5cm}}
			\toprule
			\textbf{Notation}                        & \textbf{Definition}                                                                                                                \\
			\midrule
			$J_j$                                    & the job $J_j$                                                                                                                      \\
			$J$                                      & the job set $J = \{J_1, \ldots, J_n\}$                                                                                             \\
			$p_j$                                    & the processing time of job $J_j$                                                                                                   \\
			$d_j$                                    & the due date of job $J_j$                                                                                                          \\
			$w_j$                                    & the weight of job $J_j$                                                                                                            \\
			$s$-batch                                & the keyword indicating the serial-batch machine environment                                                                        \\
			$d_j = d$                                & all jobs have a common due date $d$                                                                                                \\
			$d_j \uparrow p_j\uparrow w_j\downarrow$ & agreeable condition among due dates, processing times and weights, i.e., if $d_i\geq d_j$, then $p_i\geq p_j$ and $w_i\leq w_j$    \\
			$C_j(\sigma)$ or $C_j$                   & the completion time of job $J_j$ under schedule $\sigma$ (for simplicity, $C_j(\sigma)$ is abbreviated as $C_j$ where appropriate) \\
			$Y_j(\sigma)$ or $Y_j$                   & the late work of job $J_j$ under schedule $\sigma$, defined as $Y_j(\sigma) = \min\{\max\{C_j(\sigma) - d_j, 0\}, p_j\}$           \\
			$Y(\sigma)$ or $Y$                       & the total late work of all jobs under schedule $\sigma$, i.e., $Y(\sigma) = \sum_{j=1}^n Y_j(\sigma)$                              \\
			$Y_w(\sigma)$ or $Y_w$                   & the total weighted late work of all jobs under schedule $\sigma$, i.e., $Y_w(\sigma) = \sum_{j=1}^n w_jY_j(\sigma)$                \\
			\bottomrule
		\end{tabular}
	}
\end{table}

The above given problem $1|s\text{-}batch|Y_w$ is $NP$-hard even when all jobs have a common due date and a unit weight, in contrast to the problems with a single machine without batching, $1||Y_w$ and $1||Y$, which are polynomially solvable for $d_j=d$ \citep[cf.][]{potts1992single,hariri1995single}.

\begin{spacing}{1.48}
	\begin{theorem}\label{Th: s-batch $NP$-hardness}
		Problem $1|s\text{-}batch, d_j=d|Y$ is $NP$-hard.
	\end{theorem}

	\begin{proof}
		We use the $NP$-complete partition problem for the reduction to the decision counterpart of problem $1|s\text{-}batch, d_j=d|Y$.\\
		\noindent\textbf{Partition Problem:} Given a set of $t$ positive integers, $A=\{a_1,a_2,\ldots,a_t\}$ with $\sum_{i=1}^t a_i=2B$, can $A$ be partitioned into two sets $A_1$ and $A_2$ such that $\sum_{a_j\in A_{1}}a_{j}=\sum_{a_j\in A_{2}}a_{j}=B$?\\
		Given an instance $I$ of partition problem, we construct an instance $I'$ of the decision counterpart of problem $1|s\text{-}batch, d_j=d|Y$ as follows:

		\noindent$\bullet$ $n=t$;

		\noindent$\bullet$ $s=1$;

		\noindent$\bullet$ $p_{j}=a_{j}$, $d_j=B+1$, for $j=1,2, \ldots, n$;

		\noindent$\bullet$ The decision problem is whether there is a schedule $\pi$ such that $Y\leq B$.

		We can easily see that the decision counterpart of problem $1|s\text{-}batch, d_j=d|Y$ has a solution if and only if partition problem has a solution. Thus, the theorem holds.
	\end{proof}
\end{spacing}

Since problem $1|s\text{-}batch, d_j=d|Y$ is $NP$-hard, the more general cases with the objective of minimising the total weighted late work, $1|s\text{-}batch, d_j=d|Y_w$, with various due dates, $1|s\text{-}batch|Y$, as well as the weighted problem with various due dates, $1|s\text{-}batch|Y_w$, are also $NP$-hard. In the following section, we propose a pseudo-polynomial time algorithm solving $1|s\text{-}batch|Y_w$, which can be obviously applied for the considered special cases with a common due date, $1|s\text{-}batch, d_j=d|Y_w$,  and with an agreeable condition among due dates, processing times and weights, $1|s\text{-}batch,d_j\uparrow p_j\uparrow w_j\downarrow|Y_w$, and prove their weak $NP$-hardness.

\section{Problem \texorpdfstring{$1|s\text{-}batch|Y_w$}{1|s-batch|Yw}}\label{Sec: DP}

This section addresses the general problem $1|s\text{-}batch|Y_w$. We propose a pseudo-polynomial time algorithm and then provide an acceleration framework to enhance computational performance of this method.

\subsection{Dynamic Programming Algorithm}

The dynamic programming algorithm solving problem $1|s\text{-}batch|Y_w$ is based on a property of optimal schedules showed in Lemma \ref{Lem: s-batch general}. In the following of this paper, we call job $J_j$ $non$-$late$ if $C_j \leq d_j + p_j$, and $late$ if $C_j > d_j + p_j$. Accordingly, we call a batch $non$-$late$ $batch$ if all jobs in this batch are non-late, and $late$ $batch$ if all jobs in this batch are late.

\begin{lemma}\label{Lem: s-batch general}
	For problem $1|s\text{-}batch|Y_w$, there exists an optimal schedule satisfying conditions:\\
	(1) The late jobs are scheduled in a batch following the non-late jobs.\\
	(2) For any two non-late jobs $J_i$ and $J_j$ with $d_i\leq d_j$, if $J_j$ is processed in batch $B_k$, then $J_i$ is scheduled in batch $B_{q}$, where $q\leq k+1$.
\end{lemma}

\begin{proof}
	Since we can always shift the late jobs after the non-late jobs without increasing the objective value, Condition (1) holds.

	With regard to Condition (2), we prove it by using interchange approach. Assume that $\sigma$ is an optimal schedule violating Condition (2), i.e., as showed in Figure \ref{Fig: Fig_lemma_1.1} there exists two jobs $J_i$ and $J_j$ with  $d_i \leq d_j$ and $J_j \in B_k$ but $J_i \in B_{k+l}$, where $l \geq 2$.

	\begin{figure}[!htb]
		\centering
		\includegraphics[scale=0.8]{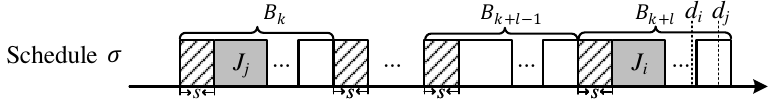}
		\caption{Optimal schedule $\sigma$ violating Condition (2)}
		\label{Fig: Fig_lemma_1.1}
	\end{figure}

	Let $\sigma'$ be a schedule after shifting job $J_j$ to batch $B_{k+l-1}$, as showed in Figure \ref{Fig: Fig_lemma_1.2}.	Since in schedule $\sigma'$ job $J_j$ is still early, i.e., $C_j(\sigma') \leq d_j$, and the completion times of other jobs do not increase, schedule $\sigma'$ is not worse than schedule $\sigma$.

	\begin{figure}[!htb]
		\centering
		\includegraphics[scale=0.8]{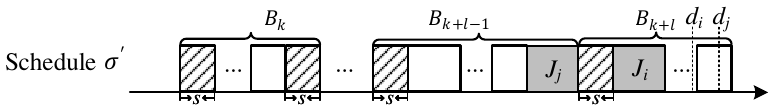}
		\caption{Optimal schedule $\sigma'$ after shifting job $J_j$ to batch $B_{k+l-1}$}
		\label{Fig: Fig_lemma_1.2}
	\end{figure}
	\vspace{-1em}
	By repeating the above argument, a schedule satisfying Condition (2) can be obtained.
\end{proof}

\begin{sloppypar}
	Based on Lemma \ref{Lem: s-batch general}, we design a DP algorithm (named $Algorithm$ $DP\text{-}F$) to solve problem $1|s\text{-}batch|Y_w$. The recursive function is defined as $f_j^b(L_1, t_1, \underline{d}, \bar{d}, L_2, t_2, k)$, the value of which is equal to the minimum total weighted late work for jobs $J^{(j)} = \{J_j, \ldots, J_n\}$ assuming that all jobs are numbered in earliest due date (EDD) order and non-late jobs from $J^{(1)}$ are processed in exactly $b$ batches. Among these batches, the first and the second non-late batches with regard to their completion times are denoted as $B^{(j)}_1$ and $B^{(j)}_2$, respectively. The completion times of these batches are equal to $t_1$ and $t_2$, while the total processing time of jobs from $J^{(j)}$ processed in them are equal to $L_1$ and $L_2$, respectively. The minimum and maximum due dates among the jobs scheduled in batch $B_1^{(j)}$ are $\underline{d}$ and $\bar{d}$, respectively. We call $J_x \in B_2^{(j)}$ $non$-$disturbed$ if $\bar{d} < d_x$, $semi$-$disturbed$ if $\bar{d} = d_x$, and $disturbed$ if $\bar{d} > d_x$. If $k=0$ then all jobs in $B_2^{(j)}$ are non-disturbed jobs. If $k=1$ then there exists at least one semi-disturbed job and does not exist any disturbed job in $B_2^{(j)}$. Otherwise, if $k=2$ then there exists at least one disturbed job in $B_2^{(j)}$. Note that batches $B^{(j)}_1$ and $B^{(j)}_2$ do not need to be executed consecutively, i.e., they may be separated by a batch consisting of jobs from $\{J_1, \ldots, J_{j-1}\}$. Table \ref{Tab: notation for Algorithm DP-F} presents a summary of the notations used in $Algorithm$ $DP\text{-}F$.
\end{sloppypar}

\begin{table}[!htbp]

	\renewcommand{\arraystretch}{1.4}
	\centering
	\caption{Notations for $Algorithm$ $DP\text{-}F$}
	\label{Tab: notation for Algorithm DP-F}
	\resizebox{\textwidth}{!}{
		\begin{tabular}{p{3.5cm} p{9.5cm}}
			\toprule
			\textbf{Notation}                                      & \textbf{Definition}                                                                                                                                                                                                                                     \\
			\midrule
			$J^{(j)}$                                              & the set $J^{(j)} = \{J_j, \ldots, J_n\}$                                                                                                                                                                                                                \\
			$b$                                                    & the number of batches used to schedule all non-late jobs from set $J$, i.e., $J^{(1)}$.                                                                                                                                                                 \\
			$B_i$                                                  & the $i$-th batch (in terms of completion time) among all batches                                                                                                                                                                                        \\
			$B^{(j)}_1$                                            & the first batch (in terms of completion time) among all non-late batches for jobs in $J^{(j)}$                                                                                                                                                          \\
			$B^{(j)}_2$                                            & the second batch (in terms of completion time) among all non-late batches for jobs in $J^{(j)}$                                                                                                                                                         \\
			$t_1$                                                  & the completion time of batch $B^{(j)}_1$                                                                                                                                                                                                                \\
			$t_2$                                                  & the completion time of batch $B^{(j)}_2$                                                                                                                                                                                                                \\
			$L_1$                                                  & the total processing time of jobs from $J^{(j)}$ processed in $B^{(j)}_1$                                                                                                                                                                               \\
			$L_2$                                                  & the total processing time of jobs from $J^{(j)}$ processed in $B^{(j)}_2$                                                                                                                                                                               \\
			$\underline{d}$                                        & the minimum due date among jobs in batch $B_1^{(j)}$                                                                                                                                                                                                    \\
			$\bar{d}$                                              & the maximum due date among jobs in batch $B_1^{(j)}$                                                                                                                                                                                                    \\
			$k$                                                    & an indicator variable where $k=0$ if all jobs in $B_2^{(j)}$ are non-disturbed; $k=1$ if there is at least one semi-disturbed job and no disturbed job; $k=2$ if there is at least one disturbed job                                                    \\
			$f_j^b(L_1, t_1, \underline{d}, \bar{d}, L_2, t_2, k)$ & the recursive function that returns the minimum total weighted late work for jobs $J^{(j)} = \{J_j, \ldots, J_n\}$ under the condition that exactly $b$ non-late batches exist and state $(L_1, t_1, \underline{d}, \bar{d}, L_2, t_2, k)$ is satisfied \\
			$f_j^0$                                                & the simplified notation for the recursive function $f_j^b(L_1, t_1, \underline{d}, \bar{d}, L_2, t_2, k)$ in the case of $b = 0$                                                                                                                        \\
			$f_j^1(L_1, t_1, \underline{d}, \bar{d})$              & the simplified notation for the recursive function $f_j^b(L_1, t_1, \underline{d}, \bar{d}, L_2, t_2, k)$ in the case of $b = 1$                                                                                                                        \\
			$p^{(j)}_{\min}$                                       & the minimum  processing time of the jobs in $J^{(j)}$, i.e., $p^{(j)}_{\min} = \min_{k = j,\ldots,n}\{p_k\}$                                                                                                                                            \\
			$P^{(j)}$                                              & the total processing time of the jobs in $J^{(j)}$, i.e., $P^{(j)} = \sum_{k=j}^n p_k$                                                                                                                                                                  \\
			$P^{(j,i)}$                                            & the sum of the first $i$ shortest processing times in $J^{(j)}$, i.e., $P^{(j,i)} = \sum_{k = j}^{j + i - 1}p'_k$ with $p'_j,\ldots,p_n'$ corresponding to non-decreasing order of $p_j,\ldots,p_n$                                                     \\
			\bottomrule
		\end{tabular}
	}
\end{table}

The $Algorithm$ $DP\text{-}F$ decides for each job $J_j$, $j=n, \ldots, 1$, whether it should be processed as a non-late or late job. The non-late jobs are scheduled backwards.
Taking into account Lemma \ref{Lem: s-batch general}, since a disturbed job can be scheduled only in the following batch, it is sufficient to consider the first two successive batches with non-late jobs while constructing a partial schedule. All late jobs are executed in a late batch at the end of the schedule. Thus, depending on the parameter values, job $J_j$ can be scheduled in five possible ways:
\begin{itemize}
	\item late;
	\item non-late and processed in $B^{(j)}_1$:
	      \begin{itemize}
		      \item as the only job from $J^{(j)}$ (then $L_1=p_j$):
		            \begin{itemize}
			            \item [$>$] causing at least one job to be a semi-disturbed job in $B_2^{(j)}$, which was  $B_1^{(j+1)}$ in the previous DP iteration performed for $J^{(j+1)}$;
			            \item [$>$] causing all job in $B_2^{(j)}$, i.e., $B_1^{(j+1)}$, to be non-disturbed jobs;
		            \end{itemize}
		      \item with at least one more job from $J^{(j)}$ (then $L_1\geq p_j + p_{\min}^{(j+1)}$):
		            \begin{itemize}
			            \item [$>$] not affecting the jobs in $B_2^{(j)}$, i.e., $B_2^{(j+1)}$, if $J_x \in B_2^{(j+1)}$ is disturbed (semi-disturbed/non-disturbed) then it is still disturbed (semi-disturbed/non-disturbed) in $B_2^{(j)}$;
		            \end{itemize}
	      \end{itemize}
	\item non-late and processed in $B^{(j)}_2$:
	      \begin{itemize}
		      \item as the only job from $J^{(j)}$ (then $L_2 = p_j$):
		            \begin{itemize}
			            \item [$>$] as a semi-disturbed job;
			            \item [$>$] as a disturbed job;
		            \end{itemize}
		      \item with at least one more job from $J^{(j)}$ (then $L_2 \geq p_j + p_{\min}^{(j+1)}$):
		            \begin{itemize}
			            \item [$>$] as a semi-disturbed job;
			            \item [$>$] as a disturbed job.
		            \end{itemize}
	      \end{itemize}
\end{itemize}

In the iterative process, while scheduling job $J_j$ we distinguish 4 possible cases, depending on the number of batches containing non-late jobs from $J^{(j)}$ in a partial schedule: $b=0$, $b=1$, $b=2$ and $3 \leq b \leq j$. Considering that for $b = 0$ there does not exist non-late batch, and for $b = 1$ there is only one non-late batch in the schedule, we simplify the parameters in the recursive function $f_j^b(L_1, t_1, \underline{d}, \bar{d}, L_2, t_2, k)$ as follows:
\begin{itemize}
	\item For $b = 0$, the recursive function is simplified to $f_j^0$ since the parameters related to batches $B^{(j)}_1$ and $B^{(j)}_2$ can be omitted in the recursive function.
	\item For $b = 1$, the recursive function is simplified to $f_j^1(L_1, t_1, \underline{d}, \bar{d})$ since there is only one non-late batch in the schedule, the parameters related to batch $B^{(j)}_2$ can be omitted in the recursive function.
	\item For $b = 2$ and $3 \leq b \leq j$, the recursive function remains $f_j^b(L_1, t_1, \underline{d}, \bar{d}, L_2, t_2, k)$.
\end{itemize}
Remark: The computation of $f_j^b(L_1, t_1, \underline{d}, \bar{d}, L_2, t_2, k)$ with $2\leq b\leq j$ may call the values of $f_j^1(L_1, t_1, \underline{d}, \bar{d})$, and the computation of $f_j^1(L_1, t_1, \underline{d}, \bar{d})$ may call the values of $f_j^0$.\\

\noindent Case 1: $b=0$ \\
For $j \in \{n,\ldots,1\}$\\
(\textit{There is no batch with non-late jobs from $J^{(j)}$, so all these jobs must be scheduled late, no job can be a disturbed job, and batches $B_1^{(j)}$ and $B_2^{(j)}$ do not exist.})
\begin{adjustwidth}{15pt}{0pt}
	$f_j^0 = f_{j+1}^0 + w_jp_j$\\
\end{adjustwidth}

\noindent Case 2: $b=1$ \\
\noindent For $j \in \{n,\ldots,1\}~\wedge~p^{(j)}_{\min} \leq L_1 \leq P^{(j)} ~\wedge~L_1 + s \leq t_1\leq ns+P^{(1)} ~\wedge~\underline{d}, \bar{d} \in \{d_j, \ldots, d_n\}$, where $\wedge$ denotes logical ``and'' used in conditional expressions\\
(\textit{There is only one batch with non-late jobs from $J^{(j)}$, so no job can be a disturbed job, this batch must contain at least one job.})\\

\begin{figure}[!htb]
	\centering
	\includegraphics[scale=0.6]{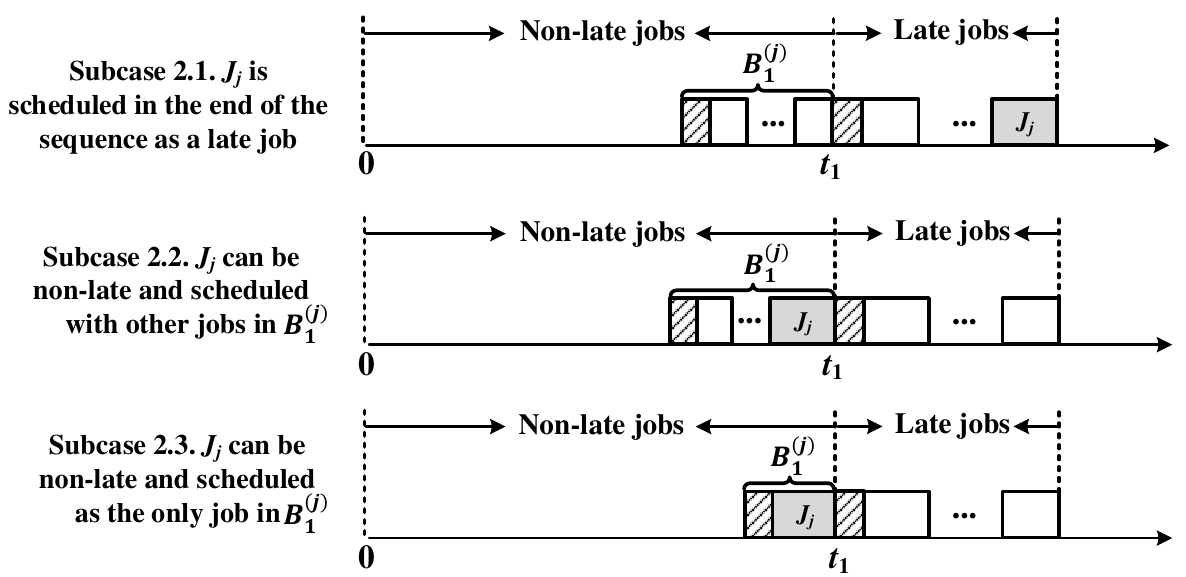}
	\caption{Subcases of Case 2}
	\label{Fig: Fig_DP_b=1}
\end{figure}

\begin{adjustwidth}{15pt}{0pt}

	\noindent Subcase 2.1. If $p_{\min}^{(j)} < L_1 <p_j ~\vee~p_j <L_1 < p_j + p_{\min}^{(j+1)}~\vee~t_1 \geq d_j + p_j$, where $\vee$ denotes logical ``or'' used in conditional expressions\\
	(\textit{$J_j$ must be scheduled as a late job. Conditions $p_{\min}^{(j)} < L_1 <p_j$ and $p_j <L_1 < p_j + p_{\min}^{(j+1)}$ imply that $J_j$ cannot be scheduled with other jobs from $J^{(j)}$ in $B_1^{(j)}$ or scheduled as the only job from this set in batch $B_1^{(j)}$. Condition $t_1 \geq d_j + p_j$ implies that $J_j$ cannot be scheduled as a non-late job.})\\
	$f_j^1(L_1, t_1, \underline{d},\bar{d}) = f_{j+1}^1(L_1, t_1, \underline{d},\bar{d}) + w_j p_j$\\

	\noindent Subcase 2.2. If $L_1 \geq p_j + p_{\min}^{(j+1)} ~\wedge~ t_1 < d_j + p_j ~\wedge~\underline{d} = d_j$\\
	(\textit{$J_j$ can be non-late and scheduled with other jobs from $J^{(j)}$ in $B^{(j)}_1$.})\\
	$f_j^1(L_1, t_1, \underline{d}, \bar{d}) = \min \{f_{j+1}^1(L_1, t_1, \underline{d}, \bar{d}) + w_jp_j,$\\
	\indent $f_{j+1}^1(L_1 - p_j, t_1, \underline{d}', \bar{d}) + w_j\max \{t_1 - d_j, 0 \}~\text{for}~\underline{d}' \in \{d_{j+1}, \ldots, d_n \}$\\
	\indent $\text{and}~\underline{d}'\leq \bar d \}$\\

	\noindent Subcase 2.3. If $L_1 = p_j ~\wedge~t_1 < d_j  + p_j ~\wedge~\underline{d} = \bar{d} = d_j$\\
	(\textit{$J_j$ can be non-late and scheduled as the only job from $J^{(j)}$ in $B^{(j)}_1$.})\\
	$f_j^1(L_1, t_1, \underline{d}, \bar{d}) = \min \{f_{j+1}^1(L_1, t_1, \underline{d}, \bar{d}) + w_j p_j,$\\
	\indent $f_{j+1}^0 + w_j \max\{t_1 - d_j,0 \}\}$\\

	\noindent Otherwise, if none of the conditions in Subcases 2.1-2.3 are satisfied, then $f_j^1(L_1, t_1, \underline{d}, \bar{d}) = \infty$\\
\end{adjustwidth}

\noindent Case 3: $b=2$ \\
\noindent For $j\in \{n - 1,\ldots,1\}~\wedge~p^{(j)}_{\min} \leq L_1 \leq P^{(j)}-p^{(j)}_{\min} ~\wedge~L_1 + s\leq t_1\leq (n-1)s + P^{(1)} - p_{\min}^{(j)} ~\wedge~\underline{d},\bar{d} \in \{d_j, \ldots, d_n\}~\wedge~p_{\min}^{(j)} \leq L_2 \leq  P^{(j)} - L_1 ~\wedge~t_1 + s + L_2 \leq t_2 \leq ns + P^{(1)}~\wedge~k\in\{0,1,2\}$\\
(\textit{There are two batches with non-late jobs from $J^{(j)}$, so all ways of scheduling $J_j$ are possible.})\\

\begin{figure}[!htb]
	\centering
	\includegraphics[scale=0.55]{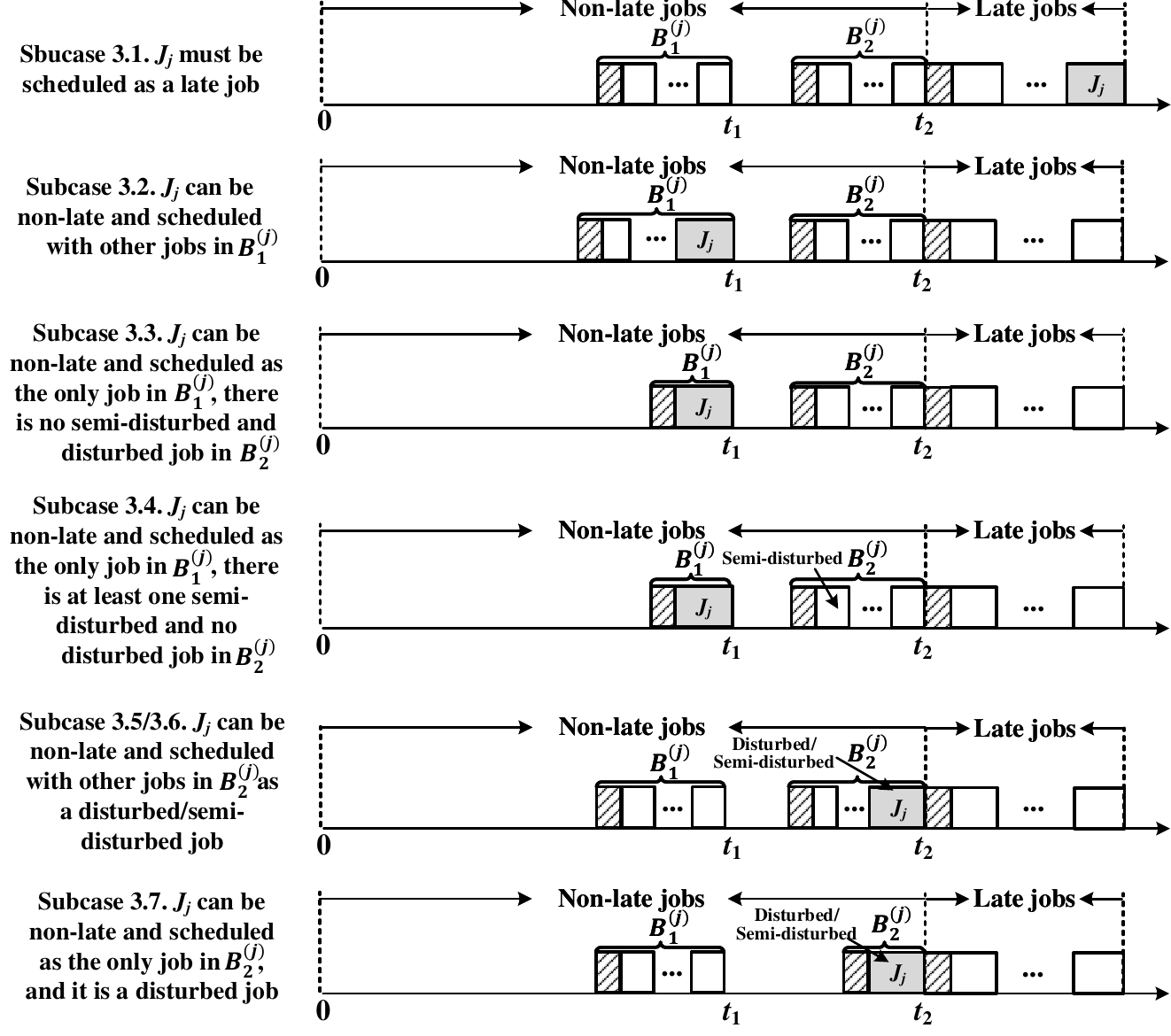}
	\caption{Subcases of Case 3}
	\label{Fig: Fig_DP_b=2}
\end{figure}

\begin{adjustwidth}{15pt}{0pt}

	\noindent Subcase 3.1. If $(p_{\min}^{(j)} < L_1 <p_j ~\vee~ p_j < L_1 < p_j + p_{\min}^{(j + 1)}~\vee~t_1 \geq d_j + p_j)~\wedge~(p_{\min}^{(j)} < L_2 <p_j ~\vee~ p_j < L_2 < p_j + p_{\min}^{(j + 1)}~\vee~t_2 \geq d_j + p_j)$\\
	(\textit{$J_j$ must be scheduled as a late job. Conditions $p_{\min}^{(j)} < L_1 <p_j$, $p_j < L_1 < p_j + p_{\min}^{(j + 1)}$, $p_{\min}^{(j)} < L_2 <p_j$, and $p_j < L_2 < p_j + p_{\min}^{(j + 1)}$ imply that $J_j$ cannot be scheduled with other jobs from $J^{(j)}$ or scheduled as the only job from this set in $B_1^{(j)}$ and $B_2^{(j)}$. Conditions $t_1 \geq d_j + p_j$ and $t_2 \geq d_j + p_j$ imply that $J_j$ cannot be scheduled as a non-late job in batches $B_1^{(j)}$ and $B_2^{(j)}$.})\\
	$f_j^2(L_1, t_1, \underline{d}, \bar{d}, L_2, t_2, k) =$
	$f_{j+1}^2(L_1, t_1, \underline{d}, \bar{d}, L_2, t_2, k) + w_j p_j$\\

	\noindent Subcase 3.2. If $L_1\geq p_j + p_{\min}^{(j+1)} ~\wedge~t_1 < d_j +p_j ~\wedge~\underline{d} = d_j$\\
	(\textit{$J_j$ can be non-late and scheduled with other jobs from $J^{(j)}$ in $B^{(j)}_1$}.)\\
	$f_j^2(L_1, t_1, \underline{d},\bar{d}, L_2, t_2, k) = \min \{f_{j+1}^2(L_1, t_1, \underline{d}, \bar{d}, L_2, t_2, k) + w_j p_j,$\\
	\indent $f_{j+1}^2(L_1 - p_j, t_1, \underline{d}', \bar{d}, L_2, t_2, k) + w_j \max\{t_1 - d_j, 0\}~\text{for}~\underline{d}' \in \{d_{j+1}, \ldots, d_n \}$\\
	\indent $\text{and}~\underline{d}'\leq \bar d \}$\\

	\noindent Subcase 3.3. If $L_1 = p_j ~\wedge~t_1 < d_j + p_j ~\wedge~\underline{d} = \bar d = d_j ~\wedge~ k = 0$\\
	(\textit{$J_j$ can be non-late and scheduled as the only job from $J^{(j)}$ in $B^{(j)}_1$, and there is no semi-disturbed and disturbed job in $B_2^{(j)}$.})\\
	$f_j^2(L_1, t_1, \underline{d}, \bar{d}, L_2, t_2, k) = \min \{f_{j+1}^2(L_1, t_1, \underline{d}, \bar{d}, L_2, t_2, k) + w_j p_j,$\\
	\indent $f_{j+1}^1(L_2, t_2, \underline{d}', \bar{d}') + w_j \max\{t_1 - d_j, 0\}~\text{for}~\underline{d}',\bar{d}' \in \{d_{j+1}, \ldots, d_n\}$\\
	\indent $\text{and}~\underline{d}'>d_j\}$\\

	\noindent Subcase 3.4. If $L_1 = p_j ~\wedge~t_1 < d_j + p_j ~\wedge~\underline{d} = \bar d = d_j ~\wedge~ k = 1$\\
	(\textit{$J_j$ can be non-late and scheduled as the only job from $J^{(j)}$ in $B^{(j)}_1$, and there is at least one semi-disturbed job and no disturbed job in $B_2^{(j)}$.})\\
	$f_j^2(L_1, t_1, \underline{d}, \bar{d}, L_2, t_2, k) = \min \{f_{j+1}^2(L_1, t_1, \underline{d}, \bar{d}, L_2, t_2, k) + w_j p_j,$\\
	\indent $f_{j+1}^1(L_2, t_2, \underline{d}', \bar{d}') + w_j \max\{t_1 - d_j, 0\}~\text{for}~\bar{d}' \in \{d_{j+1}, \ldots, d_n\}$\\
	\indent $\text{and}~ \underline{d}' = d_j \}$\\

	\noindent Subcase 3.5. If $L_2 \geq p_j + p_{\min}^{(j+1)} ~\wedge~t_2 < d_j +p_j ~\wedge~\bar d > d_j~\wedge~k = 2$\\
	(\textit{$J_j$ can be non-late and scheduled with other jobs from $J^{(j)}$ in $B^{(j)}_2$, and it is a disturbed job, i.e., $\bar d > d_j$.})\\
	$f_j^2(L_1, t_1, \underline{d}, \bar{d}, L_2, t_2, k) = \min \{f_{j+1}^2(L_1, t_1, \underline{d}, \bar{d}, L_2, t_2, k) + w_j p_j,$\\
	\indent $f_{j+1}^2(L_1, t_1, \underline{d}, \bar{d}, L_2 - p_j, t_2, k') + w_j \max\{t_2 - d_j, 0\}~\text{for}~k' = 0, 1, 2\}$\\

	\noindent Subcase 3.6. If $L_2 \geq p_j + p_{\min}^{(j+1)} ~\wedge~ t_2 < d_j + p_j ~\wedge~\bar d = d_j ~\wedge~ k = 1$\\
	(\textit{$J_j$ can be non-late and scheduled with other jobs from $J^{(j)}$ in $B^{(j)}_2$ as a semi-disturbed job, and there is no disturbed job in $B_2^{(j)}$.})\\
	$f_j^2(L_1, t_1, \underline{d}, \bar{d}, L_2, t_2, k) = \min \{f_{j+1}^2(L_1, t_1, \underline{d}, \bar{d}, L_2, t_2, k) + w_j p_j,$\\
	\indent $f_{j+1}^2(L_1, t_1, \underline{d}, \bar{d}, L_2 - p_j, t_2, k') + w_j \max\{t_2 - d_j, 0\}~\text{for}~k' = 0, 1\}$\\

	\noindent Subcase 3.7. If $L_2 = p_j ~\wedge~t_2 < d_j + p_j ~\wedge~((k =  2 ~\wedge~\bar d > d_j)~\vee~(k = 1~\wedge~ \bar d = d_j))$\\
	(\textit{$J_j$ can be non-late and scheduled as the only job from $J^{(j)}$ in $B^{(j)}_2$, and it is a disturbed job for $\bar d > d_j$, and is a semi-disturbed job for $\bar d = d_j$.})\\
	$f_j^2(L_1, t_1, \underline{d}, \bar{d}, L_2, t_2, k) = \min \{f_{j+1}^2(L_1, t_1, \underline{d}, \bar{d}, L_2, t_2, k) + w_j p_j,$\\
	\indent $f_{j+1}^1(L_1, t_1, \underline{d}, \bar{d}) + w_j \max\{t_2 - d_j,0 \}\}$\\


	\noindent Otherwise, if none of the conditions in Subcases 3.1-3.7 are satisfied, then $f_j^2(L_1, t_1, \underline{d}, \bar{d}, L_2, t_2, k) = \infty$\\

\end{adjustwidth}

\noindent Case 4: $3 \leq b \leq n-j+1$ \\
Let $p'_j, \ldots, p'_n$ be the non-decreasing order of $p_j, \ldots, p_n$. Let $P^{(j,i)} = \sum_{k = j}^{j + i - 1}p'_k$ denote the sum of first $i$ shortest processing times among jobs $J_j, \ldots, J_n$, where $1 \leq i \leq n - j + 1$.

\noindent For $j \in \{n-2,\ldots,1\}~\wedge~p^{(j)}_{\min} \leq L_1 \leq P^{(j)} - P^{(j, b - 1)} ~\wedge~L_1 + s \leq t_1 \leq (n - b + 1)s+P^{(1)}-P^{(j, b-1)}~\wedge~\underline{d}, \bar{d} \in \{d_j, \ldots, d_n\}~\wedge~p^{(j)}_{\min} \leq L_2 \leq P^{(j)} - L_1 - P^{(j, b - 2)} ~\wedge~ t_1 + s + L_2 \leq t_2 \leq (n - b + 2)s + P^{(1)} - P^{(j, b - 2)}$\\
(\textit{There are more than two batches with non-late jobs from $J^{(j)}$ which must contain at least one job, all ways of scheduling $J_j$ in first two batches are possible.})\\

\begin{figure}[!htb]
	\centering
	\includegraphics[scale=0.5]{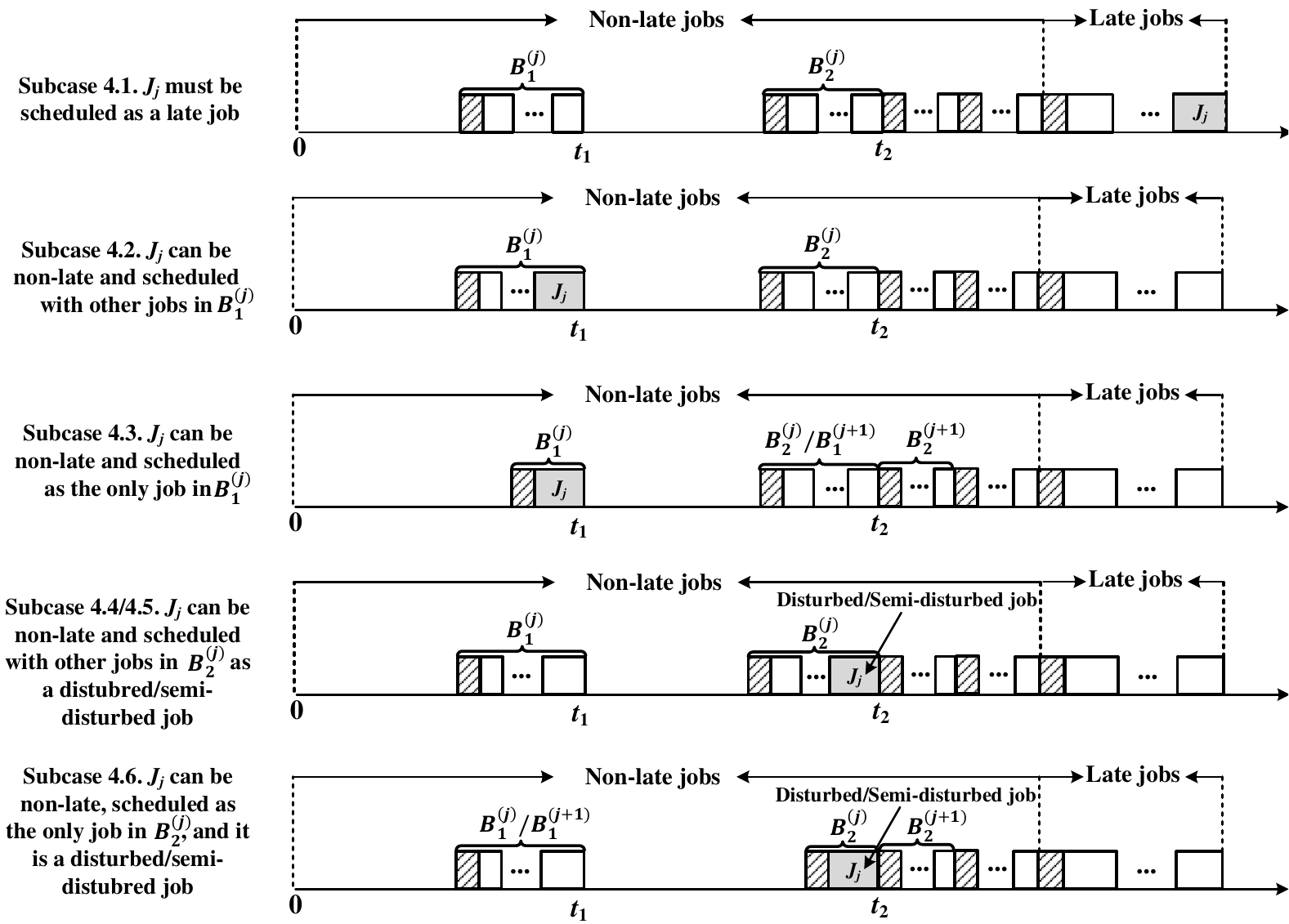}
	\caption{Subcases of Case 4}
	\label{Fig: Fig_DP_b=3}
\end{figure}

\begin{adjustwidth}{15pt}{0pt}
	\noindent Subcase 4.1. If $(p_{\min}^{(j)} < L_1 <p_j ~\vee~ p_j < L_1 < p_j + p_{\min}^{(j + 1)}~\vee~t_1 \geq d_j + p_j)~\wedge~(p_{\min}^{(j)} < L_2 <p_j ~\vee~ p_j < L_2 < p_j + p_{\min}^{(j + 1)}~\vee~t_2 \geq d_j + p_j)$\\
	(\textit{$J_j$ must be scheduled as a late job. Conditions $p_{\min}^{(j)} < L_1 <p_j$, $p_j < L_1 < p_j + p_{\min}^{(j + 1)}$, $p_{\min}^{(j)} < L_2 <p_j$, and $p_j < L_2 < p_j + p_{\min}^{(j + 1)}$ imply that $J_j$ cannot be scheduled with other jobs or scheduled as the only job from $J^{(j)}$ in $B_1^{(j)}$ and $B_2^{(j)}$. Conditions $t_1 \geq d_j + p_j$ and $t_2 \geq d_j + p_j$ imply that $J_j$ cannot be scheduled as a non-late job in batches $B_1^{(j)}$ and $B_2^{(j)}$.})\\
	$f_j^b(L_1, t_1, \underline{d}, \bar{d}, L_2, t_2, k) = f_{j+1}^b(L_1, t_1, \underline{d}, \bar{d}, L_2, t_2, k) + w_j p_j$\\

	\noindent Subcase 4.2. If $L_1 \geq p_j + p_{\min}^{(j+1)} ~\wedge~t_1 < d_j + p_j ~\wedge~\underline{d} = d_j$\\
	(\textit{$J_j$ can be non-late and scheduled with other jobs from $J^{(j)}$ in $B^{(j)}_1$.})\\
	$f_j^b(L_1, t_1, \underline{d}, \bar{d}, L_2, t_2, k) = \min \{f_{j+1}^b(L_1, t_1, \underline{d}, \bar{d}, L_2, t_2, k) + w_j p_j,$\\
	\indent $f_{j+1}^b(L_1 - p_j, t_1, \underline{d}', \bar{d}, L_2, t_2, k) + w_j \max\{t_1 - d_j, 0\} ~\text{for}~\underline{d}' \in \{d_{j+1}, \ldots, d_n \}$\\
	\indent $\text{and}~\underline{d}'\leq \bar d \}$\\

	\noindent Subcase 4.3. If $L_1=p_j ~\wedge~t_1 < d_j + p_j~\wedge~\underline{d} = \bar d = d_j~\wedge~(k = 0~\vee~k = 1)$\\
	(\textit{$J_j$ can be non-late and scheduled as the only job from $J^{(j)}$ in $B^{(j)}_1$, and there is no disturbed job in $B_2^{(j)}$. Jobs from $B^{(j)}_2$ form the first batch from the point of view of set $J^{(j+1)}$, i.e. $B^{(j+1)}_1$, the other jobs may belong to the second batch from the point of view of set $J^{(j+1)}$, i.e. $B^{(j+1)}_2$ which contains at least one job, or be processed in the following batches.})\\
	\noindent $f_j^b(L_1, t_1, \underline{d}, \bar{d}, L_2, t_2, k) = \min \{f_{j+1}^b(L_1, t_1, \underline{d}, \bar{d}, L_2, t_2, k) + w_j p_j,$\\
	\indent $\min \{f_{j+1}^{b-1}(L_2, t_2, \underline{d}', \bar{d}', L'_2, t'_2, k') + w_j \max\{t_1 - d_j,0 \}$ for $k' \in \{0, 1, 2\}~\wedge~t_2 + $\\
	\indent $s + p_{\min}^{(j + 1)} < t'_2 \leq  t_2 + s + P^{(j+1)} - P^{(j+1,b-2)} ~\wedge~L'_2 = t'_2 - t_2 - s~\wedge~ \bar{d}' \in $\\
	\indent $\{d_{j+1}, \ldots, d_n\}~\wedge~(\underline{d}' > d_j~\text{for}~k=0~\vee~\underline{d}' = d_j~\text{for}~k = 1)\}$\\
	Note that $k\in \{0, 1\}$ instead of $k\in \{0,1,2\}$ since all jobs are indexed in EDD order and job $J_j$ cannot be a disturbed job when it is the only job from $J^{(j)}$ in $B_1^{(j)}$.\\

	\noindent Subcase 4.4. If $L_2 \geq p_j + p_{\min}^{(j + 1)} ~\wedge~ t_2 < d_j + p_j ~\wedge~\bar d > d_j~\wedge~ k = 2$\\
	(\textit{$J_j$ can be non-late and scheduled with other jobs from $J^{(j)}$ in $B^{(j)}_2$ as a disturbed job, i.e., $\bar d > d_j$.})\\
	$f_j^b(L_1, t_1, \underline{d}, \bar{d}, L_2, t_2, k) = \min \{f_{j + 1}^b(L_1, t_1, \underline{d}, \bar{d}, L_2, t_2, k) + w_j p_j,$\\
	\indent \indent $f_{j + 1}^b(L_1, t_1, \underline{d}, \bar{d}, L_2 - p_j, t_2, k') + w_j \max\{t_2 - d_j,0 \} ~\text{for}~ k' = 0, 1, 2\}$\\

	\noindent Subcase 4.5. If $L_2 \geq p_j + p_{\min}^{(j + 1)} ~\wedge~t_2 < d_j + p_j ~\wedge~\bar d = d_j~\wedge~k = 1$\\
	(\textit{$J_j$ can be non-late and scheduled with other jobs from $J^{(j)}$ in $B^{(j)}_2$ as a semi-disturbed job, i.e., $\bar d = d_j$.})\\
	$f_j^b(L_1, t_1, \underline{d}, \bar{d}, L_2, t_2, k) = \min \{f_{j+1}^b(L_1, t_1, \underline{d}, \bar{d}, L_2, t_2, k) + w_j p_j,$\\
	\indent \indent $f_{j+1}^b(L_1, t_1, \underline{d}, \bar{d}, L_2 - p_j, t_2, k') + w_j \max\{t_1 - d_j, 0\} ~\text{for}~ k' = 0, 1\}$\\

	\noindent Subcase 4.6. If $L_2 = p_j ~\wedge~t_2 < d_j + p_j ~\wedge~((k = 2~\wedge~\bar d > d_j)~\vee~(k = 1~\wedge~\bar d = d_j))$\\
	(\textit{$J_j$ can be non-late, scheduled as the only job from $J^{(j)}$ in $B^{(j)}_2$, and it is a disturbed job for $\bar d > d_j$, and is a semi-disturbed job for $\bar d = d_j$. Batch $B^{(j)}_1$ forms the first batch for $J^{(j + 1)}$ and batch $B^{(j)}_2$ is inserted between the batches $B^{(j + 1)}_1$ and $B^{(j + 1)}_2$ for $J^{(j + 1)}$.})\\
	\noindent $f_j^b(L_1, t_1, \underline{d}, \bar{d}, L_2, t_2, k) = \min \{f_{j+1}^b(L_1, t_1, \underline{d}, \bar{d}, L_2, t_2, k) + w_j p_j,$\\
	\indent $f_{j+1}^{b-1}(L_1, t_1, \underline{d}, \bar{d}, L'_2, t'_2, k') + w_j \max\{t_2 - d_j, 0\} \}$\\
	\indent for $k' = 0~\wedge~t_2 + s + p_{\min}^{(j + 1)} \leq t'_2 \leq t_2 + s + P^{(j + 1)} - L_1 - P^{(j + 1, b - 2)} ~\wedge~L'_2 =$\\
	\indent $t'_2 - t_2 - s\}$\\
	Note that $k' = 0$ since there should be no disturbed and semi-disturbed job in $B_2^{(j + 1)}$ with regard to $B_1^{(j + 1)}$. By contradiction, assume that $J_i$ is a disturbed or semi-disturbed job in $B_2^{(j + 1)}$, and there exists $J_x \in  B_1^{(j + 1)}$ with $d_x \geq d_i$. After inserting batch $B_2^{(j)}$ between $B_1^{(j + 1)}$ and $B_2^{(j + 1)}$, job $J_i$ will not be scheduled immediately after $B_1^{(j + 1)}$, which contradicts Lemma \ref{Lem: s-batch general}. Hence, $k' = 0$. Additionally, $L_2' = t_2'- t_2 - s$ implies that we have to consider only those schedules that have no idle time between batches $B_2^{(j)}$ and $B_2^{(j + 1)}$.\\

	\vspace{-1em}
	\noindent Otherwise, if none of the conditions in Subcases 4.1-4.6 are satisfied, then $f_j^b(L_1, t_1, \underline{d}, \bar{d}, L_2, t_2, k) = \infty$\\
\end{adjustwidth}

\vspace{-1em}
The initial conditions for the above formulated recursive functions are given as follows:

\noindent $f_{n + 1}^0 = 0$;

\noindent $f_{n+1}^1(L_1, t_1, \underline{d}, \bar{d}) =
	\begin{cases}
		0,      & \text{if } L_1 = t_1 = \underline{d} = \bar{d} = 0; \\
		\infty, & \text{otherwise;}
	\end{cases}$

\noindent $f_{n+1}^b(L_1, t_1, \underline{d}, \bar{d}, L_2, t_2, k) =
	\begin{cases}
		0,      & \text{if } L_1 = t_1 = \underline{d} = \bar{d} = L_2 = t_2 = k = 0; \\
		\infty, & \text{otherwise}.
	\end{cases}$

The optimal criterion value is determined as the minimum from the following 4 values:\\
\noindent $f_1^0= \sum_{j=1}^n w_j p_j$;\\
\noindent $f_1^1(L_1, s + L_1, \underline{d}, \bar{d})$ for $p^{(1)}_{\min} \leq L_1 \leq P^{(1)}$ and $\underline{d}, \bar{d} \in \{d_1, \ldots, d_n\}$;\\
\noindent $f_1^2(L_1, s + L_1, \underline{d}, \bar{d}, L_2, 2s + L_1 + L_2, k)$\\
\indent for $\underline{d}, \bar{d} \in \{d_1, \ldots, d_n\}$, $p^{(1)}_{\min} \leq L_1 \leq P^{(1)} - p^{(1)}_{\min}$, $p^{(1)}_{\min} \leq L_2 \leq P^{(1)} - L_1$,
\indent and $k \in \{0, 1, 2\};$\\
\noindent $\min_{b=3, \ldots ,n} \{f_1^b (L_1, s + L_1, \underline{d}, \bar{d}, L_2, 2s + L_1 + L_2, k)~\text{for}~\underline{d}, \bar{d} \in \{d_1, \ldots, d_n\}, p^{(1)}_{\min} \leq L_1 \leq P^{(1)}-P^{(1, b - 1)}, p_{\min}^{(1)} \leq L_2 \leq P^{(1)} - L_1 - P^{(j, b - 2)}~\text{and}~k\in \{0, 1, 2\}\}.$

In the above recursive formulation, Subcases 2.2, 3.2, 4.2 involves the minimisation over $\underline{d}'$, and Subcases 4.3, 4.6 involves the minimisation over $t'_2$, $\underline{d}'$, $\bar{d}'$. These inner loops increase the state computation time if calculated directly. We develop a pre-calculation framework that reduces such transitions to $O(1)$ time. Let $\tilde{f}_{j+1}^b$ denote the auxiliary tables derived from $f_{j+1}^b$ by the minimum aggregation.

\begin{lemma}\label{Lem: Pre-calculation}
	By pre-calculating the auxiliary tables $\tilde{f}_{j+1}^b$ after the iteration for job $J_{j+1}$ and before the iteration for job $J_j$, each state transition involving inner minimisation (Subcases 2.2, 3.2, 4.2, 4.3, 4.6) can be evaluated in $O(1)$ time, and the total pre-calculation time is $O(n^4(\sum_{i=1}^n p_i)^2(ns+\sum_{i=1}^n p_i)^2)$.
\end{lemma}

\begin{proof}
	\begin{sloppypar}

		In Subcase 2.2, the recurrence requires $\min_{\underline{d}' \in \{d_{j+1}, \ldots, d_n\}, \underline{d}' \leq \bar{d}} f_{j+1}^1(L_1 - p_j, t_1, \underline{d}', \bar{d})$. We define the auxiliary table as
		$\tilde{f}_{j+1}^1(L_1, t_1, \bar{d}) = \min_{\underline{d}' \in \{d_{j+1}, \ldots, d_n\}, \underline{d}' \leq \bar{d}} \{ f_{j+1}^1(L_1, t_1, \underline{d}', \bar{d}) \},$
		where parameters satisfy the bounds defined in Subcase 2.2. For each fixed pair $(L_1, t_1)$, we compute the values for all $\bar{d} \in \{d_{j+1}, \ldots, d_n\}$ in a single pass. By iterating through due dates in EDD order and maintaining a cumulative minimum, this operation takes $O(n)$ time per $(L_1, t_1)$ pair. The total pre-calculation time for this table across all $n$ iterations is $O(n^3 (\sum_{i=1}^n p_i) (ns + \sum_{i=1}^n p_i))$.

		Similarly, in Subcases 3.2 and 4.2, we define
		$\tilde{f}_{j+1}^b(L_1, t_1, \bar{d}, L_2, t_2, k) = \min_{\underline{d}' \in \{d_{j+1}, \ldots, d_n\}, \underline{d}' \leq \bar{d}} \{ f_{j+1}^b(L_1, t_1, \underline{d}', \bar{d}, L_2, t_2, k) \},$
		where parameters satisfy the bounds in the main recurrence. For each fixed tuple $(L_1, t_1, L_2, t_2, k)$, we compute entries for all $\bar{d}$ by a single linear pass over the due dates, taking $O(n)$ time. The total pre-calculation time for this table across all $n$ iterations is $O(n^3 (\sum_{i=1}^n p_i)^2 (ns + \sum_{i=1}^n p_i)^2)$.

		For Subcase 4.3 with $k=0$, we consider the minimisation over variables $t'_2$, $\underline{d}'$, $\bar{d}'$, and $k'$ subject to $L'_2 = t'_2 - t_2 - s$ and $\underline{d}' > d_j$. We define the auxiliary function
		$\tilde{f}_{j+1}^{b-1}(L_2, t_2, \tilde{t}'_2) = \min_{t'_2 \geq \tilde{t}'_2,~ \underline{d}' > d_j,~ \bar{d}' \geq \underline{d}'} \min_{k'} \{ f_{j+1}^{b-1}(L_2, t_2, \underline{d}', \bar{d}', t'_2 - t_2 - s, t'_2, k') \},$
		where the parameters follow the bounds specified in the recurrence. The computation proceeds by a backward traversal over $t'_2$ to determine suffix minimums, followed by minimisation over the valid due date pairs $(\underline{d}', \bar{d}')$.
		Since $L'_2 = t'_2 - t_2 - s$ represents the total processing time of a batch, it follows that $t'_2 - t_2 - s \le \sum_{i=j+1}^n p_i < \sum_{i=1}^n p_i$. Thus, for any fixed $t_2$, the range for $t'_2$ is at most $\sum_{i=1}^n p_i$.
		Since this table is independent of $(L_1, t_1)$, the total pre-calculation time for this table across all $n$ iterations is $O(n^3(\sum_{i=1}^n p_i)(ns + \sum_{i=1}^n p_i)^2)$. A similar argument applies to Subcase 4.3 with $k=1$.

		For Subcase 4.6, we consider the minimisation over variable $t'_2$ subject to $L'_2 = t'_2 - t_2 - s$. We define the auxiliary function
		$\tilde{f}_{j+1}^{b-1}(L_1, t_1, \underline{d}, \bar{d}, t_2, \tilde{t}'_2) = \min_{\tilde{t}'_2 \leq t'_2 \leq t_{\max}} \{ f_{j+1}^{b-1}(L_1, t_1, \underline{d}, \bar{d}, t'_2 - t_2 - s, t'_2, 0) \},$
		where $\tilde{t}'_2$ serves as the lower bound parameter for the time variable $t'_2$, which ranges up to $t_{\max} = t_2 + s + P^{(j + 1)} - L_1 - P^{(j + 1, b - 2)}$. Since $L'_2 = t'_2 - t_2 - s$ represents the processing time of a batch, it implies $t'_2 - t_2 - s \le \sum_{i=j+1}^n p_i < \sum_{i=1}^n p_i$. Thus, for any fixed $t_2$, the valid range of $t'_2$ spans at most $\sum_{i=1}^n p_i$.  The computation proceeds by a single backward traversal over $t'_2$ for each fixed  $(L_1, t_1, \underline{d}, \bar{d}, t_2)$ to determine suffix minimums. The total pre-calculation time for this table across all $n$ iterations is $O(n^4 (\sum_{i=1}^n p_i)^2 (ns + \sum_{i=1}^n p_i)^2)$.

		Summing the construction times over all iterations, the total pre-calculation time is bounded by $O(n^4(\sum_{i=1}^n p_i)^2(ns+\sum_{i=1}^n p_i)^2)$.

	\end{sloppypar}
\end{proof}

With the pre-calculation framework established in Lemma \ref{Lem: Pre-calculation}, we determine the overall time complexity of $Algorithm$ $DP$-$F$ for solving problem $1|s\text{-}batch|Y_w$.

\begin{sloppypar}
	\begin{theorem}\label{Th: Yw}
		Problem $1|s\text{-}batch|Y_w$ can be solved by $Algorithm$ $DP\text{-}F$ in $O(n^4(\sum_{j=1}^n p_j)^2(ns+\sum_{j=1}^n p_j)^2)$ time.
	\end{theorem}
\end{sloppypar}

\begin{proof}
	The analysis given while proposing the DP algorithm guarantees its correctness. With regard to the time complexity,	since $0\leq j\leq n$, $1\leq b\leq n$, $0\leq L_1\leq \sum_{j=1}^n p_j$, $0\leq L_2\leq \sum_{j=1}^n p_j$, $\underline d\in \{d_1,\ldots,d_n\}$, $\bar d\in \{d_1,\ldots,d_n\}$, $s\leq t_1\leq t_2\leq ns + \sum_{j=1}^n p_j$, $k\in \{0,1,2\}$, the number of states is $O(n^4(\sum_{j=1}^n p_j)^2(ns+\sum_{j=1}^n p_j)^2)$. By Lemma \ref{Lem: Pre-calculation}, each state transition takes $O(1)$ time. The overall running time is the sum of the time required for computing all states and the pre-calculation time. Since the pre-calculation time is bounded by the number of states, the time complexity is $O(n^4(\sum_{j=1}^n p_j)^2(ns+\sum_{j=1}^n p_j)^2)$.
\end{proof}

\subsection{Acceleration Framework for Dynamic Programming Algorithm}
\label{Sec:DPF-Speedup-Theory}

To improve computational performance of $Algorithm$ $DP\text{-}F$, we propose an exact acceleration framework, and denote the resulting algorithm by $Algorithm$ $DP\text{-}FA$. The correctness of $Algorithm$ $DP\text{-}FA$ follows from the correctness of $Algorithm$ $DP\text{-}F$, because the framework preserves the original recurrences and uses bounding mechanisms to eliminate states that are unable to lead to an optimal solution. The acceleration relies on an early termination strategy, an initial upper bound, two lower bounds, and a transition domain reduction.

The early termination strategy is used to verify whether an optimal solution with zero late work is attainable. This is equivalent to checking whether the maximum lateness $L_{\max} \le 0$, which can be solved by the algorithm for $1|s\text{-}batch|L_{\max}$ in $O(n^2)$ time (cf. \citealp{webster1995scheduling}). If achievable, $Algorithm$ $DP\text{-}F$ terminates immediately with the zero objective value schedule.

If the early termination strategy is not effective, we derive an upper bound ($UB$) to reduce the state space. Lemma \ref{Lem: s-batch general} implies that an optimal schedule stays close to EDD order, with any job delayed by at most one batch. Given that such deviations are relatively rare, EDD-formed schedules provide near-optimal results. Furthermore, since any such schedule is feasible for the general problem, it yields a valid $UB$ that is then used to reduce the state space and tighten transition domains during the DP recursion.

During the backward state transitions (from $j=n$ down to $1$), we compute two complementary lower bounds to estimate the remaining late work for the unscheduled jobs. The first is a state-independent preemptive bound (denoted as $LB_{pmtn}$), which relaxes the batching and non-preemptive constraints for unscheduled jobs.
This bound is computed in $O(n\log n)$ time using a backward preemptive scheduling procedure that starts from the largest due date and, at each decision point $t$, greedily selects the job with the largest weight from the set of jobs satisfying $d_j \ge t$ (cf.\ \citealp{hariri1995single}).
The second is a state-dependent capacity bound ($LB_{state}$). For a DP state considering non-late batches $B_1$ (with completion time $t_1$ and length $L_1$) and $B_2$ (with completion time $t_2$ and length $L_2$), the available time before $B_1$ is at most $t_1 - s - L_1$. If the total processing time of unscheduled jobs with due dates earlier than $t_1$ exceeds this capacity, the excess becomes late work.
Multiplying this excess by the minimum weight among these jobs yields a valid lower bound. A symmetric check evaluates the capacity before $B_2$, and $LB_{state}$ takes the maximum of them.
Any DP state whose accumulated cost $c$ satisfies $c + \max\{LB_{pmtn}, LB_{state}\} > UB$ is guaranteed to be suboptimal and is eliminated from the state space.

Additionally, the upper bound $UB$ limits the range of feasible completion times $t$ for any batch containing job $J_j$. Given the accumulated cost $c$, the condition $c + w_j (t - d_j) \le UB$ implies $t \le d_j + \lfloor(UB - c)/w_j\rfloor$. This restriction narrows the search space for job $J_j$ during state transitions, reducing the computational effort.

\section{Problem \texorpdfstring{$1|s\text{-}batch, d_j = d|Y_w$}{1|s-batch, dj=d|Yw}}\label{Sec: Common Due Date}

Problem $1|s\text{-}batch, d_j = d|Y_w$ is weakly $NP$-hard according to Theorems \ref{Th: s-batch $NP$-hardness} and \ref{Th: Yw}. In the following of this section, we present an $O(n\max\{n p_{\max} d, (d + p_{\max})^2\})$-time DP algorithm based on the property of optimal solutions given in Lemma \ref{Lem: s-batch common due date}, which is faster than the method proposed for the general problem with different due dates (i.e.,  $Algorithm$ $DP\text{-}F$).

\begin{lemma}\label{Lem: s-batch common due date}
	For problem $1|s\text{-}batch,d_j = d|Y_w$, there exists an optimal schedule satisfying conditions:\\
	(1) The late jobs are scheduled in a batch following the non-late jobs.\\
	(2) There are at most two non-late batches.\\
	(3) The second non-late batch (if it exists) includes only one job, and this job is partially early.
\end{lemma}

\begin{proof}
	Condition (1) holds by using the same argument as in Lemma \ref{Lem: s-batch general}.

	Assume that there exists an optimal schedule containing $k$ non-late batches, where $k\geq 3$. Since the first $k-1$ non-late batches can be combined as one non-late batch without increasing the objective value, Condition (2) holds.

	We prove Condition (3) by contradiction. Assume that there exists an optimal schedule satisfying Conditions (1) and (2) but having more than two jobs in the second non-late batch. By shifting any job $J_j$ from the second batch to the first batch, we can obtain a new schedule. In this new schedule, the weighted late work of $J_j$ is reduced to 0 and the late work of other jobs remain unchanged. The new schedule is better than the optimal schedule, which leads to a contradiction. Hence, there is only one job in the second batch. Moreover, the job in the second batch should be partially early, otherwise the first and second batches can be combined as one batch without increasing the objective value. Thus, Condition (3) holds.
\end{proof}

By Lemma \ref{Lem: s-batch common due date}, the optimal schedule must satisfy one of the three mutually exclusive cases. We propose $Algorithm$ $DP\text{-}H$ to find the optimal objective value by using three DP functions.

\textbf{Case 1} corresponds to a schedule with a single non-late batch completing before or at $d$. Let $h_j^1(L)$ be the minimum total weighted late work of a subset of jobs from $\{J_1, \dots, J_j\}$, given that the total processing time of the jobs assigned to the fully early batch is exactly $L$. The recurrence is $h_j^1(L) = \min \{ h_{j-1}^1(L) + w_j p_j, h_{j-1}^1(L - p_j) \}$, where the first term corresponds to scheduling $J_j$ as a late job, and the second term corresponds to including $J_j$ in the early batch. The initialization is $h_0^1(0) = 0$ and $h_0^1(L) = \infty$ for $L \neq 0$. The optimal objective value for this case is $\min_{0 \le L \le d-s} \{ h_n^1(L) \}$. The time complexity is $O(nd)$.

\textbf{Case 2} corresponds to a schedule with a single non-late batch completing at $t > d$. Let $\Delta = t - d$ be the length of the batch processed after $d$ ($0 < \Delta < p_{\max}$). We define $h_j^2(L, \Delta)$ as the minimum total weighted late work for the subset $\{J_1, \dots, J_j\}$ forming a non-late batch of length $L$. The recurrence is $h_j^2(L, \Delta) = \min \{ h_{j-1}^2(L, \Delta) + w_j p_j, h_{j-1}^2(L - p_j, \Delta) + w_j \Delta \}$, where the first term corresponds to scheduling $J_j$ as a late job, and the second term corresponds to including $J_j$ (only if $p_j> \Delta$) in the non-late batch. The initialization is $h_0^2(0, \Delta) = 0$ and $h_0^2(L, \Delta) = \infty$ for $L \neq 0$.  The optimal objective value for this case is $\min \{ h_n^2(L, \Delta) \text{ for } 0 < \Delta < p_{\max} \text{ and } L = d + \Delta - s \}$. The time complexity is $O(n p_{\max} d)$.

\textbf{Case 3} involves a schedule with two non-late batches, in which the second non-late batch has only one straddling (i.e., partially early) job with respect to $d$.
To solve this case, we iterate each job $J_i$ ($1 \le i \le n$) as the candidate straddling job, incorporating a strategy consisting of a pre-calculation phase and a merging phase. In the pre-calculation phase, we partition the jobs other than $J_i$ into a prefix set $\{J_1, \dots, J_{i-1}\}$ and a suffix set $\{J_{i+1}, \dots, J_n\}$. For each set, we independently compute the minimum weighted late work, denoted as $h_{pre}(i,L_1)$ and $h_{suf}(i,L_2)$, where $L_1$ and $L_2$ are the lengths of the non-late batches for the prefix and suffix sets, respectively.
The merging phase subsequently identifies the optimal pair of lengths $(L_1, L_2)$ that minimises the objective function, subject to the straddling condition of the candidate straddling job $J_i$, i.e., $d < 2s + L_1 + p_i + L_2 < d + p_i$.

For any straddling job $J_i$, the computation of $h_{pre}(i, L_1)$ and $h_{suf}(i, L_2)$ can be reduced to Case 1. Consequently, the computational complexity for evaluating all candidates is $O(n^2 d)$.

We define the composite function $h_i^3(L_1, L)$ representing the total weighted late work when the prefix set contributes length $L_1$ to a first batch of the total length $L$. Thus, we have $
	h_i^3(L_1, L) = h_{pre}(i, L_1) + h_{suf}(i, L - L_1) + w_i(2s + L + p_i - d)
$ subject to the straddling condition $d < 2s + L + p_i < d + p_i$. The optimal objective value for this case is determined by $\min \{ h_i^3(L_1, L) ~\text{for}~1\leq i\leq n, 0< L_1 < L, d-2s-p_i<L<d-2s \}$.

The optimal solution to the problem is determined by the minimum of the optimal objective values obtained in the three cases.

\begin{theorem}\label{Th: Common due date}
	Problem $1|s\text{-}batch, d_j = d|Y_w$ can be solved in $O(n p_{\max} d)$ time for $n \leq p_{\max}$ and $O(n^2 d)$ time for $n> p_{\max}$.
\end{theorem}
\begin{proof}
	The correctness of $Algorithm$ $DP$-$H$ is guaranteed by the analysis of the three cases, which cover all possible structures of an optimal schedule. We now analyze the time complexity for each case. Case 1 can be solved in $O(nd)$ time. Case 2 requires iterating the parameter $\Delta \in [1, p_{\max}]$. For each $\Delta$, the recurrence takes $O(nd)$ time, resulting in a total complexity of $O(n p_{\max} d)$. Case 3 involves pre-calculation and merging phases. In the pre-calculation phase, computing all $h_{pre}(i, L_1)$ and $h_{suf}(i, L_2)$ for all candidate straddling jobs takes $O(n^2 d)$ time. In the merging phase, for each candidate, iterating the total first-batch length $L \in (d-2s-p_{\max}, d - 2s)$ and the prefix contribution $L_1 \in (0, L)$ requires $O(dp_{\max})$ time, summing to $O(ndp_{\max})$ over all candidates. Thus, the total complexity for Case 3 is $O(nd(n+p_{\max}))$. Therefore, the overall time complexity of the $Algorithm$ $DP$-$H$ is $O(ndp_{\max})$ for $n \leq p_{\max}$ and $O(n^2 d)$ for $n> p_{\max}$.
\end{proof}

\begin{spacing}{1.4}
	\section{Problem $1|s\text{-}batch, d_j\uparrow p_j\uparrow w_j\downarrow|Y_w$}\label{Sec: Agreeable}

	In this section, we study problem $1|s\text{-}batch, d_j\uparrow p_j\uparrow w_j\downarrow|Y_w$, where the due dates, processing times and weights are agreeable in the order $d_j \uparrow p_j\uparrow w_j\downarrow$, i.e., for any two jobs $J_i$ and $J_j$, if $d_i\geq d_j$, then $p_i\geq p_j$ and $w_i\leq w_j$. This problem is a special case of the general problem $1|s\text{-}batch|Y_w$, and remains weakly $NP$-hard according to Theorems~\ref{Th: s-batch $NP$-hardness} and \ref{Th: Yw}. However, the agreeable structure enables the derivation of useful scheduling properties that facilitate efficient algorithm design. In particular, we prove that there exists an optimal schedule in which all late jobs are grouped at the end, and batches of non-late jobs follow the order of non-decreasing due dates (see Lemma~\ref{Lem: s-batch agreeable}). Based on these structural properties, we develop a dynamic programming algorithm with running time  $O(n(\sum_{j=1}^n p_j)(ns + \sum_{j=1}^n p_j))$.

	\begin{lemma}\label{Lem: s-batch agreeable}
		For problem $1|s\text{-}batch,d_j\uparrow p_j\uparrow w_j\downarrow|Y_w$, there exists an optimal schedule satisfying conditions:\\
		(1) The late jobs are scheduled in a batch following the non-late jobs.\\
		(2) For any two non-late and adjacent batches $B_x$ and $B_{x + 1}$, if $J_i \in B_x$ and $J_j \in B_{x+1}$, then $d_i \leq d_j$.
	\end{lemma}
	\begin{proof}
		Condition (1) holds by using the same argument as in Lemma \ref{Lem: s-batch general}. With regard to Condition (2), we assume that schedule $\sigma$ is optimal and contains two non-late and adjacent batches $B_x$ and $B_{x + 1}$, in which there exist two jobs $J_i \in B_x$ and $J_j \in B_{x + 1}$ with $d_i > d_j$. Since there exists an agreeable condition among the due dates, processing times and weights of jobs, i.e., $d_j \uparrow p_j\uparrow w_j\downarrow$, we have $p_i \geq p_j$ and $w_i\leq w_j$. In schedule $\sigma$, since job $J_j$ belongs to an non-late batch $B_{x + 1}$ and $d_i > d_j$, job $J_j$ is an early or partially early job and job $J_i$ is an early job. Thus, we have that  $w_j Y_j(\sigma) = w_j(C_j(\sigma) - d_j)$ and $Y_i(\sigma) = 0$.	By interchanging the positions of jobs $J_i$ and $J_j$, we obtain a new schedule $\sigma'$ with two new batches $B_x'$ and $B_{x + 1}'$ corresponding to the previous batches $B_x$ and $B_{x + 1}$. In schedule $\sigma'$, since $p_i \geq p_j$, the completion time of batch $B_x'$ is not greater than batch $B_{x}$, while the completion time of batch $B_{x+1}'$ equal the completion time of batch $B_{x+1}$. For the jobs in $B_x'\backslash \{J_j\}$, their completion times decrease to $C_{i}(\sigma) - (p_i - p_j)$ from $C_i(\sigma)$, so their late work does not increase. For the jobs in $B_{x+1}'\backslash \{J_i\}$, their completion times remain unchanged, so their late work also remains unchanged. With regard to jobs $J_j$ and $J_i$, job $J_j$ becomes an early job and job $J_i$ becomes an early or partially early job. Accordingly, we have that $Y_j(\sigma') = 0$ and  $w_i Y_i(\sigma') = w_i(C_i(\sigma') - d_i) = w_i(C_j(\sigma) - d_i)\leq w_j(C_j(\sigma) - d_i) \leq w_j Y_j(\sigma)$. Thus, $w_j Y_j(\sigma') + w_i Y_i(\sigma') \leq w_j Y_j(\sigma) + w_i Y_i(\sigma)$.  After a finite number of repetitions of the above procedure, we can obtain an optimal schedule satisfying Condition (2). Therefore, Lemma \ref{Lem: s-batch agreeable} holds.
	\end{proof}

	Accordingly, let $g_j(L,t)$ be the minimum total weighted late work for schedules containing jobs $J^{(j)} = \{J_j, \ldots, J_n\}$ subject to the conditions that the processing and completion times of batch $B_1^{(j)}$ are $L$ and $t$, respectively. The algorithm decides for each job $J_j$, $j = n,\ldots, 1$, whether it should be processed as a non-late or late job. The non-late jobs are scheduled backwards. The late jobs are scheduled in a batch at the end of the schedule. Based on Lemma \ref{Lem: s-batch agreeable}, we only need to consider the first batch $B_1^{(j)}$. In the iterative process, job $J_j$ can be scheduled as a late job, a non-late job with other jobs from $J^{(j)}$ in batch $B_1^{(j)}$, or as a unique non-late job in batch $B_1^{(j)}$.

	We are now ready to give our DP algorithm (named $Algorithm$ $DP\text{-}G$). The initialization is as follows:
	$$g_{n+1}(L, t) =
		\begin{cases}
			0,      & \text{if}~L = 0~\wedge~s\leq t \leq ns + \sum_{j=1}^n p_j; \\
			\infty, & \text{otherwise.}
		\end{cases}$$

	The recursive function for $j = n,\ldots,1$, $L = 0,1,\ldots,\sum_{j=1}^n p_j$ and $t = s +L, \ldots,ns + \sum_{j=1}^n p_j$ is as follows:

	\vspace{-1em}
	$$g_{j}(L, t) =
		\begin{cases}
			g_{j + 1}(L, t) + w_jp_j, ~~~\text{if}~p_{\min}^{(j)} < L < p_j~\vee~p_j < L < p_j + p_{\min}^{(j+1)}~\vee~t\geq d_j + p_j;            \\
			\min\{g_{j + 1}(L - p_j, t) + w_j \max\{t - d_j, 0\}, g_{j + 1}(L, t) + w_j p_j\},                                                     \\ 	\hspace{2.8cm}\text{if}~L \geq p_j + p_{\min}^{j + 1}~\wedge~t < d_j + p_j;\\
			\min\{\min\{g_{j + 1}(L', t + s + L') + w_j \max\{t - d_j, 0\},~0\leq L' \leq \sum_{i = j + 1}^n p_i\},~g_{j + 1}(p_j, t) + w_j p_j\}, \\
			\hspace{2.8cm}\text{if}~L = p_j~\wedge~t < d_j + p_j;                                                                                  \\
			\infty, \hspace{2.3cm}\text{otherwise}.
		\end{cases}$$
\end{spacing}

\begin{spacing}{1.5}
	The four formulas in the recursive function correspond to the four cases that $J_j$ must be scheduled as a late job, $J_j$ can be non-late and scheduled with other jobs from $J^{(j)}$ in batch $B_1^{(j)}$, $J_j$ can be non-late and scheduled as the only job from this set in batch $B_1^{(j)}$, and there is no feasible schedule.

	In fact, the above recursive function can be accelerated by pre-calculate an auxiliary function $g_{j+1}^{pre}(t)$ representing the minimum cost for sub-problem $J^{(j+1)}$ starting at time $t$:
	\vspace{-1em}
	$$g_{j+1}^{pre}(t) = \min_{0 \le L' \le \sum_{k=j+1}^n p_k} \{ g_{j+1}(L', t + s + L') + w_j\max\{t-d_j,0\}\}.$$
	Constructing this table for all $t$ takes $O((\sum_{i=1}^n p_i)(ns + \sum_{i=1}^n p_i))$ time. Consequently, the state transition of $g_j(L,t)$ for $L=p_j$ is reduced to $O(1)$ time through a lookup of $g_{j+1}^{pre}(t)$.

	The optimal objective value is equal to $\min\{g_{1}(L, t): 0\leq L\leq \sum_{j=1}^n p_j, t = s + L\}$, and the corresponding optimal schedule is found by backtracking.

	\begin{theorem}\label{Th:agreeable}
		Problem $1|s\text{-}batch, d_j\uparrow p_j\uparrow w_j\downarrow|Y_w$ can be solved by $Algorithm$ $DP\text{-}G$ in $O(n(\sum_{j=1}^n p_j)(ns + \sum_{j=1}^n p_j))$ time.
	\end{theorem}
	\begin{proof}
		The correctness of $Algorithm$ $DP$-$G$ follows from Lemma~\ref{Lem: s-batch agreeable} and its recursive structure.  The time complexity is determined by the iterations over $n$ jobs. In each iteration $j$, the algorithm sequentially constructs the auxiliary table $g_{j+1}^{pre}(t)$ and updates the states $g_j(L, t)$. The construction of $g_{j+1}^{pre}(t)$ is performed by aggregating the values in $g_{j+1}(L,t)$, which takes $O((\sum_{i=1}^n p_i)(ns + \sum_{i=1}^n p_i))$ time. Following this, each state in $g_j(L, t)$ is updated. By using the pre-calculated table $g_{j+1}^{pre}(t)$, the complexity of the transition for $L=p_j$ is reduced to an $O(1)$, allowing the state-update phase to be completed in $O((\sum_{i=1}^n p_i)(ns + \sum_{i=1}^n p_i))$ time. As these two sequential procedures are performed additively for each iteration, the total running time is $O(n(\sum_{j=1}^n p_j)(ns + \sum_{j=1}^n p_j))$.
	\end{proof}

	It should be mentioned that problem $1|s\text{-}batch, d_j\uparrow p_j\uparrow w_j\uparrow|Y_w$ is not investigated in this study because it lacks a structural property that enables the design of a more efficient algorithm than $Algorithm$ $DP\text{-}F$. In some cases, prioritizing short, light-weight jobs with tight due dates may lead to suboptimal outcomes, as deferring them in favour of longer jobs with higher urgency and penalty may yield better overall performance, i.e., Property (2) in Lemma \ref{Lem: s-batch agreeable} no longer holds. Without a clear dominance structure, it becomes difficult to restrict the job ordering or simplify the decision process. Therefore, we only focus on the more structured and tractable setting of $d_j\uparrow p_j\uparrow w_j\downarrow$ in this section. Obviously, the case with $d_j\uparrow p_j\uparrow w_j\uparrow$ can be still solved by $Algorithm$ $DP\text{-}F$ proposed for the general problem $1|s\text{-}batch|Y_w$ in Section \ref{Sec: DP}.
\end{spacing}

\section{Computational Experiments}\label{Sec: Experiments}

This section evaluates the computational performance of the proposed algorithms against a benchmark Mixed-Integer Linear Programming
(MILP) model (formulated in the Appendix). The core dynamic programming algorithms were implemented in C++17. The MILP benchmark was solved using Gurobi Optimizer 13.0.1. Python 3.14.3 was used for instance generation, algorithm invocation, and result aggregation. All computational experiments were executed on an Apple Mac Studio equipped with an Apple M4 Max CPU and 36 GB of unified memory and the macOS 15.2 operating system. All benchmark instances, their corresponding optimal results, and the source code are publicly accessible as stated in Data Availability Section.

\subsection{Instance Generation}\label{Sec: Instance Generation}
To evaluate the practical applicability and the theoretical limits of our algorithms, we generated two classes of test instances: (i) \textit{Small-scale Instant O2O Delivery Instances} reflecting the general problem solved by $Algorithm$ $DP$-$F$ and its accelerated version $Algorithm$ $DP$-$FA$, and (ii) \textit{Large-scale Trunk-line Instances} featuring structured (agreeable) properties managed by $Algorithm$ $DP$-$G$ and $Algorithm$ $DP$-$H$. These two classes of instances reflect two possible applications of the considered model arising in Online-to-Offline logistics in urban on-demand services and in trunk-line optimization in high-throughput sorting centers.  Detailed parameter configurations are summarized in Table~\ref{tab:experimental_settings} (all parameters are integers).

\vspace{0.4cm}

\textbf{Class I: Instant O2O Delivery (Small-scale General Case).}
This class reflects the decision-making characteristic of urban on-demand services (Online-to-Offline logistics). The instance scale is focused on the small number of jobs $n \in \{5, 10, \dots, 30\}$ to simulate rolling-horizon dispatching. In such environments, the system optimizes small, frequent batches of orders to satisfy tight delivery. Time is measured in minutes, as pre-transportation processes such as manual picking and loading remain labor-intensive and typically operate on a minute-level granularity.

\begin{itemize}
	\item \textbf{Parameters:} Processing times $p_j$ are sampled from $U[5, 15]$ min, representing manual picking and packing. Weights $w_j$ are uniformly distributed in $U[1, 5]$ to capture diverse customer priority levels. The setup time is fixed at $s \in \{5, 10, 15\}$ min, representing constant durations for courier shift handovers or container replacement.
	\item \textbf{Temporal Constraints:} To reflect the ``delivery wave'' model, typical in urban service-level agreements, due dates are discrete: $d_j \in \{ \Delta, 2\Delta, \dots, \lceil \frac{\gamma \sum p_k}{\Delta} \rceil \Delta \}$, with $\Delta = 30$ min. The parameter $\gamma \in \{0.5, 0.7, 0.9\}$ regulates window tightness, simulating conditions from peak-hour urgency ($\gamma=0.5$) to relative off-peak flexibility ($\gamma=0.9$).
\end{itemize}

\textbf{Class II: Trunk-line Logistics (Large-scale Special Cases).}
To assess the scalability of our DP formulations for special cases, we construct large-scale instances across two magnitudes: (i) \textbf{Practical Scale} ($n \in \{100, \dots, 500\}$) and (ii) \textbf{Stress Test Scale} ($n \in \{1000, \dots, 5000\}$). Time units are scaled to seconds since trunk-line operations rely on high-speed automated sorting and conveying systems, where task durations are managed with second-level accuracy.

\begin{itemize}
	\item \textbf{Class II-A: Trunk-line with Common Due Date (Algorithm DP-H).}
	      This subset mimics high-throughput sorting centers with a shared departure due date. Setup times follow $s \sim U[300, 600]$ s, representing batch-level changeovers. Job processing times $p_j \sim U[5, 15]$ s capture rapid automated operations. Weights $w_j$ are uniformly distributed in $U[1, 5]$. All jobs share a common due date $d = \lfloor \gamma \sum p_j \rfloor$ with $\gamma \in \{0.5, 0.7, 0.9\}$.
	\item \textbf{Class II-B: Trunk-line with Agreeable Property (Algorithm DP-G).}
	      This subset reflects value-density fulfillment, where smaller parcels are typically handled faster and assigned higher priority. Job processing times $p_j \sim U[5, 15]$ s. Weights $w_j$ are uniformly distributed in $U[1, 5]$. The agreeable property is enforced by sorting independent random samples such that $p_1 \le \dots \le p_n, d_1 \le \dots \le d_n$, and $w_1 \ge \dots \ge w_n$. Due dates follow the delivery wave model in Class I with $\Delta=1800$ s and $\gamma \in \{0.5, 0.7, 0.9\}$.
\end{itemize}

\begin{table}[H]
	\centering
	\caption{Summary of experimental parameter settings}
	\label{tab:experimental_settings}
	\renewcommand{\arraystretch}{0.75}
	\small
	\begin{tabular}{l l}
		\toprule
		\textbf{Parameter}     & \textbf{Configuration}                                                                                                        \\
		\midrule
		\multicolumn{2}{l}{\textbf{Class I: Instant O2O Delivery ($Algorithm$ $DP$-$F$ $and$ $DP$-$FA$)}}                                                      \\
		\midrule
		Number of jobs $n$     & $\{5, 10, \dots, 30\}$                                                                                                        \\
		Setup times $s$        & $\{5, 10, 15\}$                                                                                                               \\
		Processing times $p_j$ & $U[5, 15]$                                                                                                                    \\
		Weights $w_j$          & $U[1, 5]$                                                                                                                     \\
		Due dates $d_j$        & $d_j \in \{ \Delta, \dots, \lceil \frac{\gamma \sum p_j}{\Delta} \rceil \Delta \}, \Delta=30, \gamma \in \{0.5, 0.7, 0.9\}$   \\
		\midrule
		\multicolumn{2}{l}{\textbf{Class II-A: Trunk-line Common Due Date ($Algorithm$ $DP$-$H$)}}                                                             \\
		\midrule
		Number of jobs $n$     & $\{100, \dots, 500\}$ and $\{1000, \dots, 5000\}$                                                                             \\
		Setup times $s$        & $U[300, 600]$                                                                                                                 \\
		Processing times $p_j$ & $U[5, 15]$                                                                                                                    \\
		Weights $w_j$          & $U[1, 5]$                                                                                                                     \\
		Common due date $d$    & $d=\lfloor \gamma \sum p_j \rfloor, \gamma \in \{0.5, 0.7, 0.9\}$                                                             \\
		\midrule
		\multicolumn{2}{l}{\textbf{Class II-B: Trunk-line Agreeable ($Algorithm$ $DP$-$G$)}}                                                                   \\
		\midrule
		Number of jobs $n$     & $\{100, \dots, 500\}$ and $\{1000, \dots, 5000\}$                                                                             \\
		Setup times $s$        & $U[300, 600]$                                                                                                                 \\
		Processing times $p_j$ & $U[5, 15]$                                                                                                                    \\
		Weights $w_j$          & $U[1, 5]$                                                                                                                     \\
		Due dates $d_j$        & $d_j \in \{ \Delta, \dots, \lceil \frac{\gamma \sum p_j}{\Delta} \rceil \Delta \}, \Delta=1800, \gamma \in \{0.5, 0.7, 0.9\}$ \\
		Sorting property       & $p_1 \le \dots \le p_n, d_1 \le \dots \le d_n, w_1 \ge \dots \ge w_n$                                                         \\
		\bottomrule
	\end{tabular}
\end{table}

\subsection{Performance Metrics}

Each experimental configuration is characterized by a parameter tuple $(n, \gamma)$.
For each unique setting, the algorithms are executed over 10 independent instances. As the proposed algorithms are exact methods, we primarily evaluate their performance in terms of computational efficiency. We report the Average CPU Time and the Average Number of States explored in the DP state space as core metrics.
\begin{itemize}
	\item \textbf{Average CPU Time (s)}: Mean execution time over 10 instances with an 1800-second time limit (unsolved instances are recorded as 1800s). The number of optimal solutions is indicated as a superscript. For example, $^{(3,10)}3.2$ denotes an average time of $3.2$ seconds with 3 of 10 instances solved to optimality.
	\item \textbf{Average Number of States}: The mean number of states generated and explored in the dynamic programming recursion, serving as a direct measure of the algorithm's computational complexity.
\end{itemize}

\subsection{Computational Results for Algorithms \texorpdfstring{$DP\text{-}F$}{DP-F} and $DP\text{-}FA$}

The computational performance of Algorithm $DP\text{-}F$ and its accelerated Algorithm $DP\text{-}FA$ is presented in Table~\ref{tab:res_dp_f}. While Gurobi remains effective for smaller instances where $n \le 10$, its computation time increases rapidly as the problem scale grows. For instance, at $n=15$ and $\gamma=0.5$, Gurobi requires an average of 855.94 seconds and fails to solve 4 out of 10 instances.
In comparison, Algorithm $DP\text{-}F$ and Algorithm $DP\text{-}FA$ solve all 10 instances in the same setting within 21.85 seconds and 4.33 seconds, respectively.
As the number of jobs increases to $n=20$ and $n=25$, Gurobi and Algorithm $DP\text{-}F$ encounter frequent timeouts,
whereas Algorithm $DP\text{-}FA$ maintains high solvability rates.

The due-date tightness $\gamma$ significantly impacts the computation times. As $\gamma$ increases, the execution time of Gurobi generally decreases, whereas the runtime of Algorithm $DP\text{-}F$ increases. In contrast, Algorithm $DP\text{-}FA$ is effective when due dates are either very tight ($\gamma=0.5$) or relatively loose ($\gamma=0.9$). At the intermediate tightness ($\gamma=0.7$), the less effective elimination of suboptimal states leads to more states being explored. Overall, the advantage of the acceleration framework is visible in the reduction of the number of explored states during the recursion. In the Algorithm $DP\text{-}F$, the state space grows to over 2 billions at $n=20$, resulting in timeouts for larger instances. In contrast, Algorithm $DP\text{-}FA$ effectively mitigates this growth by eliminating suboptimal states, maintaining a manageable state count even at $n=25$.

\begin{table}[H]
	\centering
	\scriptsize
	\setlength{\tabcolsep}{3pt}
	\caption{Computational results for $DP\text{-}F$ on Class I instances}
	\label{tab:res_dp_f}
	\begin{tabular}{cc c ccc cc}
		\toprule
		$n$                 & $\gamma$ & CPU\_Grb           & CPU\_DP-F          & CPU\_DP-FA         & States\_DP-F & States\_DP-FS \\
		\midrule
		\multirow{3}{*}{5}  & 0.5      & 0.03               & 0.18               & 0.13               & 6606         & 2241          \\
		                    & 0.7      & 0.02               & 0.01               & 0.00               & 32997        & 4051          \\
		                    & 0.9      & 0.02               & 0.01               & 0.00               & 32997        & 4051          \\
		\midrule
		\multirow{3}{*}{10} & 0.5      & 9.97               & 0.43               & 0.04               & 1563323      & 328045        \\
		                    & 0.7      & 2.61               & 1.96               & 0.05               & 7071231      & 394332        \\
		                    & 0.9      & 0.49               & 3.74               & 0.05               & 12879220     & 330910        \\
		\midrule
		\multirow{3}{*}{15} & 0.5      & $^{(6,10)}855.94$  & 21.85              & 4.33               & 63110798     & 16904840      \\
		                    & 0.7      & $^{(7,10)}614.34$  & 66.54              & 5.46               & 186811396    & 20461906      \\
		                    & 0.9      & 79.91              & 153.41             & 4.63               & 416463822    & 19006305      \\
		\midrule
		\multirow{3}{*}{20} & 0.5      & $^{(2,10)}1620.70$ & 378.62             & 115.15             & 880560544    & 316912514     \\
		                    & 0.7      & $^{(2,10)}1490.79$ & $^{(3,10)}1533.82$ & 221.60             & 1795530907   & 547188660     \\
		                    & 0.9      & $^{(8,10)}586.93$  & $^{(2,10)}1663.73$ & 101.49             & 2057690834   & 260070816     \\
		\midrule
		\multirow{3}{*}{25} & 0.5      & --                 & --                 & $^{(6,10)}1071.10$ & --           & 1253755538    \\
		                    & 0.7      & $^{(1,10)}1791.27$ & --                 & $^{(1,10)}1662.96$ & --           & 1024974145    \\
		                    & 0.9      & $^{(3,10)}1502.49$ & --                 & $^{(6,10)}1102.30$ & --           & 1222014013    \\
		\midrule
		\multirow{3}{*}{30} & 0.5      & --                 & --                 & --                 & --           & --            \\
		                    & 0.7      & --                 & --                 & --                 & --           & --            \\
		                    & 0.9      & --                 & --                 & --                 & --           & --            \\
		\bottomrule
	\end{tabular}
\end{table}

\subsection{Computational Results for Algorithm \texorpdfstring{$DP\text{-}H$}{DP-H}}
The computational performance of Algorithm $DP\text{-}H$ on Class II-A instances is presented in Table~\ref{tab:res_dp_h}. For the instances with $n \le 500$, Algorithm $DP\text{-}H$ finds optimal schedules in less than 0.02 seconds on average, whereas Gurobi requires up to 1.01 seconds. As the problem scale increases to $n=5000$, this performance gap widens significantly. For massive instances with $\gamma=0.9$, Gurobi requires an average of 565.05 seconds, while $DP\text{-}H$ converges in only 1.22 seconds. This acceleration demonstrates the algorithm's practical viability for large-scale operations.

The experimental results also reveal the impact of the due-date tightness $\gamma$ on the DP algorithm's computational complexity. While the computation time and the number of explored states increase with the job scale $n$, they are also sensitive to $\gamma$. Relaxing the common due date (i.e., increasing $\gamma$ from 0.5 to 0.9) widens the time window for the non-late batch. This results in a larger state space during the recursion.
For example, at $n=5000$, increasing $\gamma$ from 0.5 to 0.9 expands the state space by approximately 1.3 times,
and the average CPU time increases from 0.77 seconds to 1.22 seconds.
Although this increase in state complexity, Algorithm $DP\text{-}H$ maintains its computational efficiency.

\begin{table}[H]
	\centering
	\footnotesize
	\caption{Computational results for $DP\text{-}H$ on Class II-A instances}
	\label{tab:res_dp_h}
	\resizebox{\textwidth}{!}{
		\setlength{\tabcolsep}{3pt}
		\begin{tabular}{c ccc ccc ccc}
			\toprule
			     & \multicolumn{3}{c}{$\gamma=0.5$} & \multicolumn{3}{c}{$\gamma=0.7$} & \multicolumn{3}{c}{$\gamma=0.9$}                                                                             \\
			\cmidrule(lr){2-4} \cmidrule(lr){5-7} \cmidrule(lr){8-10}
			$n$  & CPU\_Grb                         & CPU\_DP-H                        & States\_DP-H                     & CPU\_Grb & CPU\_DP-H & States\_DP-H & CPU\_Grb & CPU\_DP-H & States\_DP-H \\
			\midrule
			100  & 0.02                             & 0.00                             & 126121                           & 0.03     & 0.00      & 252010       & 0.02     & 0.00      & 356430       \\
			200  & 0.07                             & 0.00                             & 1053006                          & 0.07     & 0.00      & 1481280      & 0.07     & 0.00      & 1731017      \\
			300  & 0.17                             & 0.00                             & 2781091                          & 0.15     & 0.00      & 3624565      & 0.25     & 0.01      & 4081787      \\
			400  & 0.40                             & 0.01                             & 5267763                          & 0.43     & 0.01      & 6665515      & 0.52     & 0.01      & 7393238      \\
			500  & 0.63                             & 0.01                             & 8526642                          & 0.64     & 0.01      & 10619872     & 1.01     & 0.01      & 11680024     \\
			\midrule
			1000 & 5.27                             & 0.03                             & 36284917                         & 5.36     & 0.04      & 43916685     & 6.17     & 0.08      & 47514171     \\
			2000 & 29.36                            & 0.12                             & 148947608                        & 31.80    & 0.16      & 177928508    & 43.19    & 0.19      & 191018089    \\
			3000 & 93.11                            & 0.27                             & 337738636                        & 122.81   & 0.36      & 401698556    & 132.64   & 0.43      & 430233555    \\
			4000 & 199.87                           & 0.49                             & 602274469                        & 271.49   & 0.65      & 714944198    & 320.11   & 0.77      & 764884490    \\
			5000 & 394.21                           & 0.77                             & 941099306                        & 507.24   & 1.01      & 1115883927   & 565.05   & 1.22      & 1193075653   \\
			\bottomrule
		\end{tabular}
	}
\end{table}

\subsection{Computational Results for Algorithm $DP\text{-}G$}

The computational performance of Algorithm $DP\text{-}G$ on Class II-B instances is presented in Table~\ref{tab:res_dp_g}. For the instances with $n \le 500$, Algorithm $DP\text{-}G$ finds optimal schedules within a few seconds on average (e.g., 2.48 seconds for $n=500$ and $\gamma=0.7$), whereas Gurobi frequently encounters timeouts for $n>300$. As the problem scale increases to $n=4000$, Algorithm $DP\text{-}G$ maintains high solvability across all tested tightness levels, with only one timeout observed at $\gamma=0.9$ (i.e., 9 out of 10 instances solved). For the largest instances with $n=5000$ and $\gamma=0.5$, Algorithm $DP\text{-}G$ requires an average of 1272.23 seconds, while Gurobi is entirely intractable.

The experimental results also reveal the impact of the due-date tightness $\gamma$ on the DP algorithm's computational complexity. While the computation time and the number of explored states increase with the job scale $n$, they are also sensitive to $\gamma$. Relaxing the due date expands the time window for feasible completion times. This results in a larger state space during the recursion.
For example, at $n=4000$, increasing $\gamma$ from 0.5 to 0.9 expands the state space by approximately 2.3 times,
and the average CPU time of Algorithm $DP\text{-}G$ increases from 730.97 seconds to 1736.97 seconds.
Despite this increase in state complexity, Algorithm $DP\text{-}G$ maintains its computational efficiency.

\begin{table}[H]
	\centering
	\footnotesize
	\caption{Computational results for $DP\text{-}G$ on Class II-B instances}
	\label{tab:res_dp_g}
	\resizebox{\textwidth}{!}{
		\setlength{\tabcolsep}{3pt}
		\begin{tabular}{c ccc ccc ccc}
			\toprule
			     & \multicolumn{3}{c}{$\gamma=0.5$} & \multicolumn{3}{c}{$\gamma=0.7$} & \multicolumn{3}{c}{$\gamma=0.9$}                                                                                                           \\
			\cmidrule(lr){2-4} \cmidrule(lr){5-7} \cmidrule(lr){8-10}
			$n$  & CPU\_Grb                         & CPU\_DP-G                        & States\_DP-G                     & CPU\_Grb           & CPU\_DP-G & States\_DP-G & CPU\_Grb           & CPU\_DP-G          & States\_DP-G  \\
			\midrule
			100  & 1.81                             & 0.22                             & 53874121                         & 1.85               & 0.13      & 53874121     & 1.79               & 0.13               & 53874121      \\
			200  & 716.14                           & 0.36                             & 145867199                        & 715.72             & 0.36      & 145867199    & 470.55             & 0.59               & 244470353     \\
			300  & $^{(1,10)}1730.81$               & 0.60                             & 235047826                        & $^{(3,10)}1627.31$ & 2.30      & 959045624    & $^{(3,10)}1625.05$ & 2.34               & 959045624     \\
			400  & --                               & 3.63                             & 1509747139                       & --                 & 3.60      & 1509747139   & --                 & 3.98               & 1611094907    \\
			500  & --                               & 2.48                             & 2041031886                       & --                 & 2.48      & 2041031886   & --                 & 5.23               & 3701544350    \\
			\midrule
			1000 & --                               & 13.64                            & 10434478972                      & --                 & 21.52     & 15859613681  & --                 & 32.16              & 24355585929   \\
			2000 & --                               & 97.82                            & 79012200350                      & --                 & 150.75    & 124558508220 & --                 & 225.17             & 181899577530  \\
			3000 & --                               & 311.55                           & 259628214000                     & --                 & 495.20    & 398743482569 & --                 & 737.64             & 597740184113  \\
			4000 & --                               & 730.97                           & 594257089820                     & --                 & 1176.77   & 956088077546 & --                 & $^{(9,10)}1736.97$ & 1396279389026 \\
			5000 & --                               & 1272.23                          & 1035873556413                    & --                 & --        & --           & --                 & --                 & --            \\
			\bottomrule
		\end{tabular}
	}
\end{table}

\section{Conclusions and Future Research}\label{Sec: Conclusion and Future Research}

In the paper, we have studied the serial-batch scheduling problem with late work related criteria for the first time. We focused on the basic single batch machine model, which appeared to be $NP$-hard even for unit weights and a common due date, $1|s\text{-}batch, d_j = d|Y$. The serial-batching setting made the late work minimisation intractable, since the classical single machine model with a common due date was polynomially solvable even for the various job weights, $1|d_j = d|Y_w$ (cf. \citealp{hariri1995single}). The $NP$-hardness of $1|s\text{-}batch, d_j = d|Y$ implies the $NP$-hardness of more complex models, particularly $1|s\text{-}batch|Y_w$. To determine its computational complexity, we designed the dynamic programming algorithm with the pseudo-polynomial running time. The existence of this approach allowed us to classify $1|s\text{-}batch|Y_w$ as weakly $NP$-hard. The approach proposed for this model can be obviously applied to its special intractable cases. Based on the special features of optimal solutions, we proposed dedicated - more time efficient - dynamic programming methods for problems $1|s\text{-}batch, d_j = d|Y_w$ and $1|s\text{-}batch, d_j\uparrow p_j\uparrow w_j\downarrow|Y_w$. Computational experiments performed for test instances of various characteristics showed the efficiency of our algorithms, which significantly outperformed the commercial solver Gurobi.
The theoretical results obtained in the paper are summarized in Table \ref{Tab: conclusion}.
\vspace{-1em}
\begin{table}[H]
	\centering
	\caption{New results}\label{Tab: conclusion}
	\scalebox{0.7}{
		\begin{tabular}{cccc}
			\toprule
			\textbf{Problem}                                              & \textbf{Complexity}                                                    & \textbf{Reference}                \\ \midrule
			$1|s\text{-}batch|Y$                                          & weakly $NP$-hard,  $O(n^4(\sum_{j=1}^n p_j)^2(ns+\sum_{j=1}^n p_j)^2)$ & Theorem \ref{Th: Yw}              \\
			$1|s\text{-}batch|Y_w$                                        & weakly $NP$-hard,  $O(n^4(\sum_{j=1}^n p_j)^2(ns+\sum_{j=1}^n p_j)^2)$ & Theorem \ref{Th: Yw}              \\
			$1|s\text{-}batch, d_j = d|Y$                                 & weakly $NP$-hard,  $O(\max\{ndp_{\max}, n^2 d\})$                      & Theorem \ref{Th: Common due date} \\
			$1|s\text{-}batch, d_j = d|Y_w$                               & weakly $NP$-hard,  $O(\max\{ndp_{\max}, n^2 d\})$                      & Theorem \ref{Th: Common due date} \\
			$1|s\text{-}batch, d_j\uparrow p_j\uparrow w_j\downarrow|Y_w$ & weakly $NP$-hard,  $O(n(\sum_{j=1}^n p_j)(ns + \sum_{j=1}^n p_j))$     & Theorem \ref{Th:agreeable}        \\
			\bottomrule
		\end{tabular}}
\end{table}

With regard to the fact that even special cases of late work minimisation on the single serial-batch machine are $NP$-hard our future research will focus on designing approximation algorithms, preferably approximation schemes. Moreover, we will investigate serial-batch scheduling problems to minimise late work related criteria for parallel machines. The theoretical complexity studies will be completed with designing heuristic and meta-heuristic approaches, and validating their efficiency in computational experiments.

\bibliography{mybibfile}

\begin{spacing}{1.35}

	\section*{Appendix: Mathematical formulation}

	We formulate problem $1|s\text{-}batch|Y_w$ as a mixed-integer linear program (MILP). Let $K = n$ be the maximum number of batches (achieved when each job forms its own batch).

	We introduce the following scheduling decision variables:
	\begin{itemize}
		\item $x_{jk} \in \{0, 1\}$: equals 1 if job $J_j$ is assigned to batch $B_k$, and 0 otherwise;
		\item $u_k \in \{0, 1\}$: equals 1 if batch $B_k$ is non-empty, and 0 otherwise.
	\end{itemize}

	Additionally, we use the following auxiliary variables to determine completion times and late work:
	\begin{itemize}
		\item $C_k \geq 0$: the completion time of batch $B_k$;
		\item $Y_j \geq 0$: the late work of job $J_j$;
		\item $z_j \in \{0, 1\}$: equals 1 if job $J_j$ is fully late (i.e., $Y_j = p_j$), and 0 otherwise.
	\end{itemize}

	The MILP formulation is as follows:
	\begin{align}
		\min \quad & \sum_{j=1}^{n} w_j Y_j \label{eq:obj}                                                                                                   \\
		\text{s.t.} \quad
		           & \sum_{k=1}^{K} x_{jk} = 1,                                   &  & \forall j = 1, \ldots, n, \label{eq:assign}                           \\
		           & u_k \geq x_{jk},                                             &  & \forall j = 1, \ldots, n,\ \forall k = 1, \ldots, K, \label{eq:usage} \\
		           & u_k \geq u_{k+1},                                            &  & \forall k = 1, \ldots, K-1, \label{eq:symmetry}                       \\
		           & C_1 = s \cdot u_1 + \sum_{j=1}^{n} p_j x_{j1}, \label{eq:C1}                                                                            \\
		           & C_k = C_{k-1} + s \cdot u_k + \sum_{j=1}^{n} p_j x_{jk},     &  & \forall k = 2, \ldots, K, \label{eq:Ck}                               \\
		           & Y_j \geq C_k - d_j - M z_j - M(1 - x_{jk}),                  &  & \forall j = 1, \ldots, n, k = 1, \ldots, K, \label{eq:latework_lb}    \\
		           & Y_j \leq p_j,                                                &  & \forall j = 1, \ldots, n, \label{eq:Ybound}                           \\
		           & x_{jk}, u_k, z_j \in \{0, 1\}, \quad C_k, Y_j \geq 0,        &  & \forall j = 1, \ldots, n, k = 1, \ldots, K, \label{eq:domain}
	\end{align}
	where $M$ is a sufficiently large constant.

	The objective function \eqref{eq:obj} minimises the total weighted late work. Constraints \eqref{eq:assign}--\eqref{eq:Ck} define assignment of jobs to batches, usage of consecutive batches, and completion times of the first and following batches. Constraints \eqref{eq:latework_lb} and \eqref{eq:Ybound} define the late work $Y_j = \min\{p_j, \max\{0, C_j - d_j\}\}$ using the auxiliary binary variable $z_j$, which indicates whether the job is fully late. Specifically, if $0 \leq C_k - d_j \leq p_j$, then $z_j$ is set to zero reducing Constraints \eqref{eq:latework_lb} to $Y_j \geq C_k - d_j$ for batch $k$ to which job $j$ is assigned, which implies $Y_j = C_k - d_j$ under the minimisation objective. Conversely, if $C_k - d_j > p_j$, setting $z_j=0$ would violate the upper bound of late work defined by Constraint \eqref{eq:Ybound}, forcing $z_j=1$. In this case, Constraints \eqref{eq:latework_lb} becomes redundant, and Constraints \eqref{eq:Ybound} combined with the minimisation objective ensures $Y_j = p_j$.

	The proposed MILP formulation is general and naturally covers special cases. For the problem with a common due date, $1|s\text{-}batch, d_j = d|Y_w$, the parameter $d_j$ in Constraint \eqref{eq:latework_lb} is simply replaced by the constant $d$. For the problem satisfying the agreeable condition $1|s\text{-}batch, d_j \uparrow p_j \uparrow w_j \downarrow|Y_w$ (where if $d_i \leq d_j$ then $p_i \leq p_j, w_i \geq w_j$), the presented formulation remains valid.

\end{spacing}
\end{document}